%% file: main.tex
\documentclass[a4paper, UKenglish,cleveref, autoref, thm-restate, dvipsnames]{lipics-v2021}
\usepackage[utf8]{inputenc}
\usepackage{hyperref}
\usepackage{url}
\usepackage{todonotes}
\usepackage{xargs}
\usepackage{amsmath, amsthm, amssymb}
\usepackage{xcolor}
\usepackage{tikz}
\usetikzlibrary{automata,calc, positioning, graphs, decorations.pathmorphing,shapes,arrows.meta,arrows,shapes.misc, fit, math}

\numberwithin{theorem}{section}
\numberwithin{definition}{section}
\numberwithin{lemma}{section}
\numberwithin{corollary}{section}


\nolinenumbers
\hideLIPIcs
\authorrunning{Isja Mannens and Jesper Nederlof}
\author{Isja Mannens}{Utrecht University, The Netherlands}{i.m.e.mannens@uu.nl}{https://orcid.org/0000-0003-2295-0827}{}

\author{Jesper Nederlof}{Utrecht University, The Netherlands}{j.nederlof@uu.nl}{https://orcid.org/0000-0003-1848-0076}{}

\acknowledgements{Both authors are supported by the project CRACKNP that has received funding from the European Research Council (ERC) under the European Union’s Horizon 2020 research and innovation programme (grant agreement No 853234).}

\DeclareMathOperator{\ctw}{ctw}
\DeclareMathOperator{\pw}{pw}
\DeclareMathOperator{\tw}{tw}

\DeclareMathOperator{\spa}{span}
\DeclareMathOperator{\supp}{supp}
\DeclareMathOperator{\rank}{rank}

\title{A Fine-Grained Classification of the Complexity of Evaluating the Tutte Polynomial on Integer Points Parameterized by Treewidth and Cutwidth
}
\titlerunning{Tutte Polynomial Parameterized by Width Measures}
\date{\today}

\newcommandx{\jesper}[2][1=]{\todo[inline,linecolor=red,backgroundcolor=red!25,bordercolor=red,caption={\normalsize \textbf{Jesper}},#1]{\normalsize #2}}
\newcommandx{\isja}[2][1=]{\todo[linecolor=green,backgroundcolor=green!25,bordercolor=green,caption={\normalsize \textbf{Isja}},#1]{\normalsize \textbf{Isja:} #2}}

\begin{document}
\maketitle

\begin{abstract}
We give a fine-grained classification of evaluating the Tutte polynomial $T(G;x,y)$ on all integer points on graphs with small treewidth and cutwidth.
Specifically, we show for any point $(x,y) \in \mathbb{Z}^2$ that either
\begin{itemize}
    \item $T(G; x, y)$ can be computed in polynomial time,
    \item $T(G; x, y)$ can be computed in $2^{O(tw)}n^{O(1)}$ time, but not in $2^{o(ctw)}n^{O(1)}$ time assuming the Exponential Time Hypothesis (ETH),
    \item $T(G; x, y)$ can be computed in $2^{O(tw \log tw)}n^{O(1)}$ time, but not in $2^{o(ctw \log ctw)}n^{O(1)}$ time assuming the ETH,
\end{itemize}
where we assume tree decompositions of treewidth $tw$ and cutwidth decompositions of cutwidth $ctw$ are given as input along with the input graph on $n$ vertices and point $(x,y)$.

To obtain these results, we refine the existing reductions that were instrumental for the seminal dichotomy by Jaeger, Welsh and Vertigan~[Math. Proc. Cambridge Philos. Soc'90].
One of our technical contributions is a new rank bound of a matrix that indicates whether the union of two forests is a forest itself, which we use to show that the number of forests of a graph can be counted in $2^{O(tw)}n^{O(1)}$ time.
\end{abstract}

\input{Introduction}

\input{Preliminaries}

\input{ReduceAlongCurve}

\input{Forests}

\section{Other cases}

In this section we handle the remaining cases mentioned in Theorem \ref{thm:main}.

\subsection{The curve $H_2$}
The curve $H_2$ is equivalent to the partition function of the Ising model. Both our proofs for the upper and lower bound on the complexity will make use of this fact.
\input{H2Lowerbound}

\input{H2Upperbound}

\subsection{The curve $H_0^x$}
The curve $H_0^x$ contains the point $(1,2)$, which counts the number of connected edgesets of a connected graph. Using existing results this gives an \textsc{ETH} lower bound which matches the running time of the general algorithm described in Theorem \ref{thm:GenAlg}.
\input{H0XLowerbound}

By Theorem \ref{thm:GenAlg} points on this curve can be computed in time $\tw(G)^{O(\tw(G))} n^{O(1)}$.

\subsection{The curve $H_q$ for $q \in \mathbb{Z}_{\geq 3}$}
These curves contain the points $(1-q, 0)$, which count the number of $q$-colorings. Using previous results and a folklore algorithm, we find matching upper and lower bounds for these points and thus for the whole curves.
\input{HqLowerbound}

\input{HqUpperbound}

\subsection{The curve $H_{-q}$ for $q \in \mathbb{Z}_{> 0}$}
These curves contain the points $(1 + q, 0)$. Using the same results we used to prove theorem \ref{thm:HqLow} and exploiting the fact these results hold for modular counting, we find an \textsc{ETH} lower bound which matches the running time of the general algorithm described in theorem \ref{thm:GenAlg}.
\input{H-qLowerbound}

By Theorem \ref{thm:GenAlg} points on this curve can be computed in time $\tw(G)^{O(\tw(G))} n^{O(1)}$.

\input{GeneralAlgorithm}

\input{Conclusion}

\bibliography{TutteCtWRef.bib}
\end{document}

%% file: Introduction.tex
\section{Introduction}



We study the parameterized complexity of computing the Tutte Polynomial. The Tutte polynomial is a graph invariant that generalizes any graph invariant that satisfies a linear deletion-contraction recursion. Such invariants include the chromatic, flow and Jones polynomials, as well as invariants that count structures such as the number of forests or the number of spanning subgraphs. Due to its generality the Tutte polynomial is of great interest to a variety of fields, including knot theory, statistical physics and combinatorics.

For a number of these fields it is important to understand how difficult it is to compute the Tutte polynomial. A series of papers, culminating in the work by Jaeger, Vertigan, and Welsh \cite{jaeger_vertigan_welsh_1990} has given a complete dichotomy showing that the problem of evaluating the Tutte polynmial is \#P-hard on all points except on the following \emph{special points} on which it is known to be computable in polynomial time:
\begin{equation}\label{eq:special}
(1, 1), (-1, -1), (0, -1), (-1, 0), (i, -i), (-i, i), (j, j^2), (j^2, j),\ H_1
\end{equation}
where $j = e^{2\pi i/3}$ and $i=\sqrt{-1}$, and $H_\alpha$ denotes the hyperbola $\{(x,y): (x-1)(y-1) = \alpha\}$. 
These hyperbolic curves turn out to be of great importance to understanding the complexity of the Tutte Polynomial, as the problem is generally equally hard on all points of the same curve, except for the \emph{special points} listed in~\eqref{eq:special}.

Further refinements of the result by~\cite{jaeger_vertigan_welsh_1990} have since been made: Among others, a more fine-grained examination of the complexity was done by Brand et al.~\cite{BrandDR19} (building on earlier work by Dell et. al.~\cite{Dell_2014}): they showed that for almost all points the Tutte polynomial cannot be evaluated in $2^{o(n)}$ time on $n$-vertex graphs, assuming (a weaker counting version of) the Exponential Time Hypothesis.
This is tight because, on the positive side, Bj\"orklund et al.~\cite{BjorklundHKK08} showed that the Tutte polynomial can be evaluated on any point in $2^nn^{O(1)}$ time.
%

Another perspective worth examining is that of the parameterized complexity of the problem, when parameterized by \emph{width measures}. This is a rapidly evolving field within parameterized complexity.\footnote{For example, the biennial Workshop on Graph Classes, Optimization, and Width Parameters (GROW) already had its 10'th edition recently~\url{https://conferences.famnit.upr.si/event/22/}.} Intuitively, it is concerned with the effects of structural properties of the given input graph on its complexity. This often generates results that have greater practical value and give a deeper understanding of the problem, in comparison with classical worst-case analysis. It is therefore natural to ask what a complexity classification for the Tutte Polynomial would look like in this parameterized context. 

For the specific subject of evaluating the Tutte polynomial parameterized by width measures, research has already been done in this area over twenty years ago: Noble \cite{DBLP:journals/cpc/Noble98} has given a polynomial time algorithm for evaluation the Tutte Polynomial on bounded treewidth graphs. Noble mostly focused on the dependence on the number of vertices and edges, and showed each point of the Tutte polynomial can be evaluated in linear time, assuming the treewidth of the graph is constant. See also an independently discovered (but slower) algorithm by Andrzejak~\cite{Andrzejak98}.
However, this glances over the \emph{exponential} part of the runtime, i.e. the dependence on the treewidth. Since this is typically the bottleneck, recent work aims to refine our understanding of this exponential dependence with upper and lower bounds on complexity of the problem in terms of this parameter that match in a fine-grained sense.

In this work, we extend this research line and determine the fine-grained complexity for each integer point $(x,y)$ of the problem of evaluating the Tutte polynomial $(x,y)$.
As was done in previous works, we base our lower bounds on the Exponential Time Hypothesis (\textsc{ETH}) and the Strong Exponential Time Hypothesis (\textsc{SETH}) formulated by Impagliazzo and Paturi~\cite{ImpagliazzoP01}. For a given width parameter $k$, the former will be used to exclude run times of the form $k^{o(k)}n^{O(1)}$, while the latter will be used to exclude run times of the form $(c-\epsilon)^kn^{O(1)}$ for some constant $c$ and any $\epsilon >0$.


Specifically we consider the \emph{treewidth}, \emph{pathwidth} and \emph{cutwidth} of the graph. The first two, in some sense, measure how close the graph is to looking like a tree or path respectively. The cutwidth measures how many edges are layered on top of each other when the vertices are placed in any linear order. We will more precisely define these parameters in the preliminaries.

Width measures in particular are interesting because instances where such structural parameters are small come up a lot in practice. For example, the curve $H_2$ corresponds to the partition function of the Ising model, which is widely studied in statistical physics, on graphs with particular topology such as lattice graphs or open/closed Cayley trees (\cite{KRIZAN1983230}). In all such graphs with $n$ vertices, even the cutwidth (the largest parameter we study) is at most $O(\sqrt{n})$.

\subsection{Our contributions}

\begin{figure}
    \centering
        \begin{tikzpicture}[scale=0.6]
        \foreach\i in {-3,...,4}{
            \draw[gray!40,very thin] (\i,-2.2) -- (\i,4.2);
        }
        \foreach\i in {-2,...,4}{
            \draw[gray!40,very thin] (-3.2, \i) -- (4.2, \i);
        }
    
    
        \draw[thick,-{Latex[length=2mm, width=2mm]}] (-3.2, 0) -- (4.4, 0) node[right] {$x$};
        \draw[thick,-{Latex[length=2mm, width=2mm]}] (0, -2.4) -- (0, 4.4) node[above] {$y$};
        \draw[thick,domain=-2.2:4.2, smooth, variable=\x, red] plot ({1}, {\x}); 
        \draw[thick,domain=-3.2:4.2, smooth, variable=\x, blue] plot ({\x}, {1}); 
        \draw[thick,domain=-3.2:.687, smooth, variable=\x, Emerald]  plot ({\x}, {1/(\x-1) + 1 }); 
        \draw[thick,domain=1.314:4.2, smooth, variable=\x, Emerald]  plot ({\x}, {1/(\x-1) + 1 }); 
        \draw[thick,domain=-3.2:.372, smooth, variable=\x, blue]  plot ({\x}, {2/(\x-1) + 1 }); 
        \draw[thick,domain=1.628:4.2, smooth, variable=\x, blue]  plot ({\x}, {2/(\x-1) + 1 }); 
        \draw[thick,domain=-3.2:.05, smooth, variable=\x, blue]  plot ({\x}, {3/(\x-1) + 1 }); 
        \draw[thick,domain=1.95:4.2, smooth, variable=\x, blue]  plot ({\x}, {3/(\x-1) + 1 }); 
        \draw[thick,domain=-3.2:.687, smooth, variable=\x, red]  plot ({\x}, {-1/(\x-1) + 1 }); 
        \draw[thick,domain=1.314:4.2, smooth, variable=\x, red]  plot ({\x}, {-1/(\x-1) + 1 }); 
        \draw[thick,domain=-3.2:.372, smooth, variable=\x, red]  plot ({\x}, {-2/(\x-1) + 1 }); 
        \draw[thick,domain=1.628:4.2, smooth, variable=\x, red]  plot ({\x}, {-2/(\x-1) + 1 }); 
        \draw[thick,domain=-3.2:.05, smooth, variable=\x, red]  plot ({\x}, {-3/(\x-1) + 1 }); 
        \draw[thick,domain=1.95:4.2, smooth, variable=\x, red]  plot ({\x}, {-3/(\x-1) + 1 }); 
    
        \node[fill=Emerald,circle,inner sep=0.1em] at (1,1) {};
        \node[fill=Emerald,circle,inner sep=0.1em] at (-1,-1) {};
        \node[fill=Emerald,circle,inner sep=0.1em] at (0,-1) {};
        \node[fill=Emerald,circle,inner sep=0.1em] at (-1,0) {};
    \end{tikzpicture}
    \caption{The \textcolor{red}{red} points have time complexity of the form $k^{O(k)}$, the \textcolor{blue}{blue} points have time complexity of the form $O(c^k)$ for some constant $c$ and the \textcolor{Emerald}{green} points have polynomial time complexity.}
    \label{fig:ResultPlot}
\end{figure}
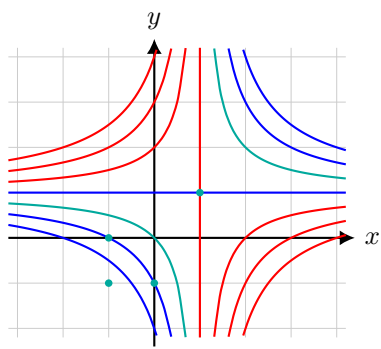

Our classification handles points $(x,y)$ differently based on whether $(x-1)(y-1)$ is negative, zero or positive, and reads as follows:
\newcommand{\itemone}{\textcolor{gray}{\textbf{1.}}}
\newcommand{\itemtwo}{\textcolor{gray}{\textbf{2.}}}
\newcommand{\itemthree}{\textcolor{gray}{\textbf{3.}}}

\begin{theorem} \label{thm:main}
Let $G$ be a graph with given tree, path and cut decompositions of width $\tw$, $\pw$ and $\ctw$ respectively. Let $(x,y) \in \mathbb{Z}^2$ be a non-special point, then up to some polynomial factor in $|G|$, the following holds.
\begin{enumerate}
    \item If $(x-1)(y-1) < 0$ or $x = 1$, then $T(G; x, y)$ can be computed in time $\tw^{O(\tw)}$ and cannot be computed in time $\ctw^{o(\ctw)}$ under \textsc{ETH}.
    \item If $y = 1$, then $T(G; x, y)$ can be computed in time $O(4^{\pw})$ or $O(64^{\tw})$ and cannot be computed in time $2^{o(\ctw)}$ under \textsc{ETH}.
    \item If $(x-1)(y-1) = q > 1$, then $T(G; x, y)$ can be computed in time $O(q^{\tw})$. Furthermore, 
    \begin{enumerate}
    \item  if $x \neq 0$, then $T(G; x, y)$ cannot be computed in time $O((q-\epsilon)^{\ctw})$ under \textsc{SETH}.
    \item  if $x=0$, then $T(G; x, y)$ cannot be computed in time $O((q-\epsilon)^{\pw})$ and $O(q-\epsilon)^{\ctw/2})$ under \textsc{SETH}.
    \end{enumerate}
\end{enumerate}
\end{theorem}

This is a fine-grained classification for evaluating the Tutte polynomial at any given integer point, simultaneously for all the parameters treewidth, pathwidth and cutwidth. This is because if a graph has cutwidth $\ctw$, pathwidth $\pw$ and treewidth $\tw$, then $\tw \leq \pw \leq \ctw$. Our result implies that, for evaluating the Tutte polynomial at a given integer point, it does not give a substantial advantage to have small cutwidth instead of small treewidth. This is somewhat surprising since, for example, for computing the closely related chromatic number of a graph there exists a $2^{\ctw} n^{O(1)}$ time algorithm, but any $\pw^{o(\pw)} n^{O(1)}$ time algorithm would contradict the ETH~\cite{LokshtanovMS18}.


Of particular interest are the upper bounds in Case \itemtwo{} for the points $\{(x,y): y = 1\}$, which are closely related to the problem of computing the number of forests in the input graph. One reason why this results stands out in particular is that it indicates an inherent asymmetry between the $x$- and $y$-axes, in this parameterized setting. In the general setting, problems related to the Tutte Polynomial often have a natural dual problem, which one can obtain by interchanging the $x$- and $y$-coordinates. For example the chromatic polynomial can be found (up to some computable term $f$) as $\chi_G(\lambda) = f(\lambda) T(1-\lambda, 0)$, while the flow polynomial can be found as $C_G(\lambda) = g(\lambda) T(0, 1-\lambda)$. These two problems are equivalent on planar graphs, in the sense that the chromatic number of a planar graph is equal to the flow number of its dual graph.

We note that for this curve we have an \textsc{ETH} bound, while for the other results of the form $c^{\tw}n^{O(1)}$ we have a stronger \textsc{SETH} bound. We suspect that a $(4-\epsilon)^{\ctw}n^{O(1)}$ lower bound for any $\epsilon >0$, based on SETH, also holds for evaluating $T(G; 2, 1)$, but that it will take significant additional technical effort.

\subparagraph{Techniques} In order to get the classification, our first step follows the method of~\cite{jaeger_vertigan_welsh_1990} to reduce the evaluation of $T(G; x, y)$ for all points in hyperbola ${H_\alpha = \{ (x,y) : (x-1)(y-1)=\alpha \}}$ to the evaluation to a \emph{single} point in $H_\alpha$. This is achieved in~\cite{jaeger_vertigan_welsh_1990} by some graph operations (\emph{stretch} and \emph{thickening}), but these may increase the involved width parameters. We refine these operations in Section~\ref{sec:redcurve} to avoid this.

With this step being made, several cases of Theorem~\ref{thm:main} then follow from a combination of new short separate and non-trivial arguments and previous work (including some very recent work such as~\cite{2015CurtMarx, 2022GNMS}).

However, for the upper bound in Case \itemtwo{} of Theorem~\ref{thm:main}, our proof is more involved. To get our upper bound, we introduced the \emph{forest compatibility matrix}. Its rows and columns are indexed with forests (encoded as partitions indicating their connected components). An entry in this matrix indicates whether the union of the two forests forms a forest itself. This matrix is closely related to matrices playing a crucial role in the Cut and Count method~\cite{talg/CyganNPPRW22} and rank based method~\cite{BodlaenderCKN15} to quickly solve connectivity problems on graphs with small tree-width. However, the previous rank upper bounds do not work for bounding the rank of the forest compatibility matrix over the reals since we check for \emph{acyclicity} instead of \emph{connectivity}. We nevertheless show that this the rank of this matrix is $4^n$; in fact the set of non-crossing partitions forms a basis of this matrix. We prove this via an inductive argument that is somewhat similar to the rank bound of $2^{n/2-1}$ of the matchings connectivity matrix over $GF(2)$ from~\cite{CyganKN18}.  Subsequently, we show how to use this insight to get a $2^{O(\tw)}$ algorithm to evaluate $T(G; 1,2)$ (i.e. counting the number of spanning forests). 

\subsection{Organization}
The remainder of this paper supports Theorem~\ref{thm:main}. In Section~\ref{sec:prel} we describe some preliminaries.
In Section~\ref{sec:redcurve} we show how to reduce the task of computing all points along a hyperbola curve to a single point.
We now describe where each part of Theorem~\ref{thm:main} can be found in the paper.
The lower bound in Case \itemone{}  is given in Theorem~\ref{thm:H-qLow} and~\ref{thm:H0Low}.
The upper bound in Case \itemone{} is given in Theorem~\ref{thm:GenAlg}.
The lower bound in Case \itemtwo{} is by Dell et al.~\cite{Dell_2014}.
The upper bound in Case \itemtwo{} is given in Section~\ref{sec:forests} (specifically, Theorems~\ref{thm:ForPWAlg} and~\ref{thm:ForTWAlg}).
%
The lower bound in Case \itemthree{} is given in Theorem~\ref{thm:H2Low} (for $q=2$) and Theorem~\ref{thm:HqLow} (for $q>2$).
The upper bound bound in \itemthree{} is given in Theorem~\ref{thm:H2Upp} (for $q=2$) and Theorem~\ref{thm:HqUpp} (for $q>2$).

%% file: Preliminaries.tex
\section{Preliminaries}
\label{sec:prel}
\subparagraph*{Computational Model}
In this paper we frequently have real (and some intermediate lemma's are even stated for complex) numbers as intermediate results of computations. However, as is common in this area we work in the word RAM model in which all basic arithmetic operations with such numbers can be done in constant time, and therefore this does not influence our running time bounds.

\subparagraph*{Interpolation}
Throughout this paper we will use interpolation to derive a polynomial, given a finite set of evaluations of said polynomial. For our purposes it suffices to note that this can be done in polynomial time, for example by solving the system of linear equations given by the Vandermonde matrix and the evaluations (see e.g.~\cite[Section 30.1]{DBLP:books/daglib/0023376}).

\begin{lemma} \label{lem:IntPol}
    Given pairs $(x_0, y_0), \dots, (x_d, y_d)$, there exists an algorithm which computes the unique degree $d$ polynomial $p$ such that $p(x_i) = y_i$ for $i = 0, \dots, d$ and runs in time $O(d^3)$.
\end{lemma}

\subsection{The Tutte polynomial}
There are multiple ways of defining the Tutte polynomial. In this paper we will only need the following definition
\[
T(G; x,y) = \sum\limits_{A \subseteq E} (x-1)^{k(A) - k(E)} (y-1)^{k(A) + |A| - |V|},
\]
where $k(A)$ denotes the number of connected components of the graph $(V,A)$.
We will often use the following notation
\[
    H_\alpha = \{(x,y) : (x-1)(y-1) = \alpha\}.
\]
Note that these curves form hyperbolas and that for $\alpha = 0$ the hyperbola collapses into two orthogonal, straight lines. We refer to these two lines as separate curves
\begin{align*}
    H_0^x &= \{(x,y) : x = 1 \}, \\
    H_0^y &= \{(x,y) : y = 1 \}.
\end{align*}

Throughout the paper we will refer to the problem of finding the value of $T(G; a,b)$ for an individual point as \emph{computing the Tutte polynomial on $(a,b)$}. We will often restrict the Tutte polynomial to a one-dimensional curve $H_\alpha$. Note that in this case the polynomial can be expressed as a univariate polynomial 
\[
    T_\alpha(G;t) := T\left(G; \frac{\alpha}{t} + 1, t + 1\right).
\]
We will refer to the problem of finding the coefficients of $T_\alpha$ as \emph{computing the Tutte polynomial along $H_\alpha$}.

\hypertarget{trg:H1Pol}{As mentioned in the introduction, the Tutte polynomial is known to be computable in polynomial time on the points
\begin{equation}\label{eq:spec}
(1,1), (-1,-1), (0,-1), (-1,0), (i, -i), (-i, i), (j, j^2), (j^2, j)
\end{equation}
and along the curve $H_1$ and it is $\#P$ to evaluate it on any other point. We call the points listed in~\eqref{eq:spec}, along with the points on the curve $H_1$ \emph{special points}. See \cite{jaeger_vertigan_welsh_1990} for more details.}

\subsection{Width measures}
We consider the width measures \emph{treewidth}, \emph{pathwidth} and \emph{cutwidth} of a graph $G$ (denoted respectively with $tw(G)$, $pw(G)$ and $ctw(G)$), defined as follows:

\subparagraph*{Treewidth and pathwidth}

A \emph{tree decomposition} of a graph is given by a tree $\mathbb{T}$ and a \emph{bag}
 $B_x \subseteq V$ for each $x \in V(\mathbb{T})$, with the following properties.
\begin{itemize}
    \item For every $v \in V(G)$ there is some $x$ such that $v \in B_x$.
    \item For every $uv \in E(G)$ there is some $x$ such that $u,v \in B_x$.
    \item For every $v \in V(G)$, the set $\{x \in V(\mathbb{T}) : v \in B_x\}$ induces a subtree of $\mathbb{T}$.
\end{itemize}
The \emph{width} of such a decomposition is defined as $\max_x(|B_x|) - 1$ and the \emph{treewidth} of a graph is defined as the minimum width among its tree decompositions. The \emph{pathwidth} of a graph is defined in a similar way, except $\mathbb{T}$ has to be a path instead of a tree.

We will often think of these decompositions as being rooted at some node $r$ and will refer to the neighbour $x$ of $y$ on the unique path from $y$ to $r$ as the parent of $y$ and to $y$ as the child of $x$. We will refer to set of nodes $y$ whose unique $y$-$r$ path visits $x$ as the descendants\footnote{Note that under this definition, $x$ is a descendant of itself.} of $x$. The union of the bags corresponding to descendants of $x$ will be denoted $G_x$ and we will refer to it as the part of the graph $G$ that \emph{lies below} $x$.

Finally we note that we may assume that we are given a so called \emph{nice} decomposition. A nice tree decomposition contains only the following types of bags.
\begin{itemize}
    \item Leaf bag: $B_x = \emptyset$ and $x$ has no children.
    \item Node-forget bag: $B_x = B_y \setminus \{v\}$, where $y$ is the unique child of $x$ and $v \in B_y$.
    \item Node-introduce bag: $B_x = B_y \cup \{v\}$, where $y$ is the unique child of $x$ and $v \in V(G) \setminus B_y$.
    \item Join bag: $B_x = B_{y_1} = B_{y_2}$, where $y_1$ and $y_2$ are the two children of $x$.
\end{itemize}
We may also assume that the decomposition has so called \emph{edge-introduce bags}. These have a unique child with the same bag, however they are labeled with an edge between two vertices in its bag. The idea behind this is that we can pretend like an edge doesn't exist, until it gets introduced by some bag. We will assume that edges are always introduced exactly once.

\subparagraph*{Cutwidth}
A \emph{cut decomposition} of a graph $G$ is simply an ordering $v_1, \dots, v_n$ of the vertex set. The width of such a decomposition is defined as the maximum number of edges 'crossing' any cut of the ordering. Formally for an integer $i$ we say an edge $v_jv_l$ crosses the $i$-th cut, if $j \leq i < l$. Again, the \emph{cutwidth} of $G$ is the minimum width among all cut decompositions.

\subsection{Brylawski's tensor product formula}
In Section \ref{sec:redcurve} we will make use of Brylawski's tensor product formula \cite{Brylawski1980} to reduce the computation of $T(G;x,y)$ to that of $T(G;x',y')$ for some other point $(x', y')$. The original formula is formulated in terms of pointed matroids, however we will only need the formulation for (multi)graphs. Before we can state the formula, we first need to introduce some notation.

Given graphs $G$ and $H$, where an edge $e \in E(H)$ is labeled as a special edge, we define the \emph{pointed tensor product}\footnote{Note that this is different from the standard tensor product for graphs.} $G \otimes_e H$ of $G$ and $H$ as the graph given by the following procedure. For every edge $f \in E(G)$ we first create a copy $H_{f}$ of $H$, then identify $f$ with the copy of the edge $e$ in $H_{f}$ and finally remove the edge f (and thus also the edge $e$) from the graph.

Intuitively it might be easier to think of this product as replacing every edge of $G$ with a copy $H \setminus e$, where two of the vertices in $H$ are designated as gluing points. For example one could replace every edge with a path of length $k$ by taking as $H$ the cycle $C_{k+1}$ on $k+1$ vertices, as seen in figure~\ref{fig:pointTensProd}.

\begin{figure}[h]
    \centering
    \begin{tikzpicture}[-, scale = 1.5]
		\tikzstyle{every state} = [inner sep = 0.1mm, minimum size = 5mm]
        \tikzstyle{vertex} = [inner sep = 0.1mm, minimum size = 2mm, draw, circle]
		
		\node[state]	(U)	    at (0,0)	    {$u$};
		\node[state]	(V)	    at (1,1)		{$v$};
		\node[state]	(W)	    at (2,0)		{$w$};

        \node[state]	(U2)	    at (4,0)	    {$u$};
		\node[vertex]	(U3)	    at (4.5,0.5)	{};
		\node[state]	(V2)	    at (5,1)		{$v$};
		\node[vertex]	(W3)	    at (5.5,0.5)	{};
		\node[state]	(W2)	    at (6,0)		{$w$};

        \draw[->, thick]   (2.5,0.7) -- (3.5,0.7);
		
		\path	(W)	    edge        node	    {}	(V)
				(U)     edge        node	    {}	(V)
                (W2)	edge        node	    {}	(W3)
                (W3)	edge        node	    {}	(V2)
				(U2)    edge        node	    {}	(U3)
				(U3)    edge        node	    {}	(V2);
				
		\end{tikzpicture}
    \caption{The pointed tensor product of the left-hand graph with a 3-cycle is given by the right-hand graph.}
    \label{fig:pointTensProd}
\end{figure}
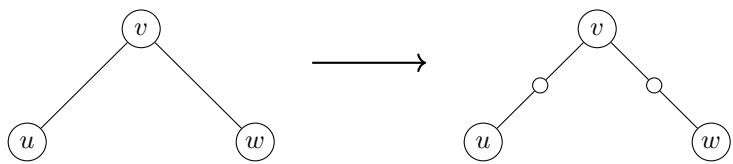

Note that this is not always well-defined, as one can choose which endpoint is identified with which. It turns out that this choice does not affect the graphic matroid of $G \otimes_e H$ and thus it does not affect the resulting Tutte polynomial. In this paper we will only consider graphs $H$ that are symmetric over $e$ and thus the product is actually well-defined.

We are now ready to state Brylawski's tensor product formula. Let $T_C$ and $T_L$ be the unique polynomials that satisfy the following system of equations
\begin{align*}
    (x-1)T_C(H; x, y) + T_L(H; x, y) &= T(H\setminus e; x, y) \\
    T_C(H; x, y) + (y-1)T_L(H; x, y) &= T(H/ e; x, y).
\end{align*}

We define
\begin{align*}
    x' &= \frac{T(H\backslash e; x, y)}{T_L(H; x, y)}
    &y' = \frac{T(H/ e; x, y)}{T_C(H; x, y)}.
\end{align*}

Let $n = |V(H)|$, $m = |E(H)|$ and $k = k(E(H))$. Brylawski's tensor product formula states that
\[
T(G \otimes_e H; x, y) = T_C(H; x, y)^{m - n + k} T_L(H; x, y)^{n - k} T(G; x', y').
\]

%% file: ReduceAlongCurve.tex
\section{Reducing along the curve $H_\alpha$}
\label{sec:redcurve}

In this section we describe how we can lift hardness results from a single point $(a, b) \in H_\alpha$ to the whole curve $H_\alpha$. We summarize the results from this section in the following theorem.

\begin{theorem}\label{thm:redCurveMain}
Let $(a,b) \in \mathbb{C}^2$. Also let $T(G; x, y)$ be the Tutte polynomial of $G$ and $\alpha := (a-1)(b-1)$. There exists a polynomial time reduction from computing $T$ on $(a,b)$ for graphs of given tree-, path- or cutwidth, to computing $T$ along $H_\alpha$ for graphs with the following width parameters.
\begin{itemize}
    \item If $|a| \notin \{0,1\}$ or if $|b| \notin \{0,1\}$ and $a \neq 0$, then the treewidth remains $\tw(G)$. The cutwidth and pathwidth become at most $\ctw(G) + 2$ and $\pw(G) + 2$ respectively.
    \item If $|b| \notin \{0,1\}$ and $a = 0$, then the treewidth remains $\tw(G)$. The pathwidth becomes at most $\pw(G) + 2$ and the cutwidth becomes at most $2\ctw(G)$.
    \item If $|a|, |b| \in \{0,1\}$, then the treewidth remains $\tw(G)$. The pathwidth becomes at most $\pw(G) + 2$ and the cutwidth becomes at most $12\ctw(G)$.
\end{itemize}
\end{theorem}

Theorem \ref{thm:redCurveMain} lets us lift both algorithms and lower bounds from a point $(a,b)$ to the whole curve $H_\alpha$. While our main theorem only requires Theorem \ref{thm:redCurveMain} to be stated for integer valued points, we will state it as the most general version we can prove. We note that for Case \itemone{} of Theorem \ref{thm:main}, we do not care too much about constant multiplicative factors in the cutwidth, since we have an ETH bound of the form $\ctw(G)^{o(\ctw(G))}$. For Case \itemtwo{} we only need the bounds on the treewidth and pathwidth. Thus the blowup in the cutwidth is only relevant for Case \itemthree{}. In this case the only integer valued points that fall under the third item of Theorem \ref{thm:redCurveMain} are $(-1, 0), (0,-1)$ and $(-1,-1)$. These are all special points, which means that this item is not relevant for Case \itemthree{}.

In our proofs we will make use of the following transformations.
\begin{definition}[\cite{jaeger_vertigan_welsh_1990}]
    Let $G$ be a simple graph. We define the \emph{$k$-stretch} $^kG$ of $G$ as the graph obtained by replacing every edge by a path of length $k$. We define the \emph{$k$-thickening} $_kG$ of $G$ as the graph obtained by replacing every edge by $k$ parallel edges.
\end{definition}    

A new variant we introduce to keep the cutwidth low is defined as follows:
\begin{definition}
    We define the \emph{insulated $k$-thickening} $_{(k)}G$ as the graph obtained by replacing every edge by a path of length $3$ and then replacing the middle edge in each of these paths by $k$ parallel edges.
\end{definition}

\begin{figure}[h]
    \centering
    \begin{tikzpicture}[-, scale = 1.5]
		\tikzstyle{every state} = [inner sep = 0.1mm, minimum size = 5mm]
		
		\node[state]	(U)	    at (0,0)	    {$u$};
		\node[state]	(W1)	at (1,0)		{};
		\node[state]	(W2)	at (3,0)		{};
		\node[state]	(V)	    at (4,0)		{$v$};
		
		\path	(W1)	edge            	    node	    {}	(U)
						edge[in=150,out=30]	    node  	    {}	(W2)
						edge[in=170,out=10]     node	    {}	(W2)
						edge[in=210,out=-30]	node	    {}	(W2)
						edge[in=190,out=-10]	node	    {}	(W2)
				(W2)    edge	                node	    {}	(V);
				
		\end{tikzpicture}
    \caption{The result of applying the insulated $4$-thickening to an edge between $u$ and $v$.}
    \label{fig:kBundle}
\end{figure}
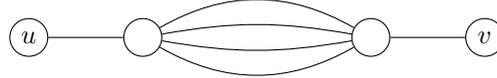

\subsection{Effect on width parameters}
We now give three lemmas that show how these transformations effect the parameters we use.
\begin{lemma} \label{lem:twPres}
     Let $G$ be a graph. Then we have that $\tw(^kG) \leq \tw(G)$, $\tw(_kG) \leq \tw(G)$ and $\tw(_{(k)}G) \leq \tw(G)$.
\end{lemma}
\begin{proof}
    First note that parallel edges do not affect the treewidth of a graph, since any bag covering one of these edges will necessarily cover all of them. This means that the original tree decomposition is also a tree decomposition for the $k$-thickening $_kG$. It also means that for the purposes of finding a tree decomposition, the insulated $k$-thickening is equivalent to a $3$-stretch. It remains to show that the $k$-stretch does not increase the treewidth.

    Note that $\tw(G) = 1$ if and only if $G$ is a tree. Since the $k$-stretch of a tree is also a tree, we find $\tw(^kG) = 1$.

    Now suppose that $\tw(G) \geq 2$. We will show that subdividing an edge does not affect the treewidth of the graph. By repeatedly subdividing edges we then find that the treewidth of the $k$-stretch $^kG$ is at most that of $G$.

    Let $uv \in E(G)$ and let $G'$ be the graph obtained by subdividing $uv$ into $uw$ and $wv$. Let $x$ be some node in the tree decomposition of $G$ such that $u,v \in B_x$. We create a tree decomposition of $G'$ by adding a node $x'$, with a corresponding bag $B_{x'} = \{u, v, w\}$, and connecting $x'$ to $x$. It is easy to see that the resulting decomposition is still a tree decomposition.
\end{proof}

\begin{lemma} \label{lem:pwPres}
     Let $G$ be a graph. Then we have that $\pw(^kG) \leq \pw(G) + 2$, $\pw(_kG) \leq \pw(G)$ and $\pw(_{(k)}G) \leq \pw(G) + 2$.
\end{lemma}
\begin{proof}
    Like in the previous proof, we note that parallel edges do not affect the pathwidth of a graph. This means that the original path decomposition is also a path decomposition for the $k$-thickening $_kG$. Again, this also means that for the purposes of finding a path decomposition, the insulated $k$-thickening is equivalent to a $3$-stretch. It remains to show that the $k$-stretch does not increase the pathwidth by more than an additive factor of 2.

    Suppose we are given a path decomposition of $G$, of width $\pw(G)$. Whenever a new vertex $v$ is introduced in a bag $B_x$, in the path decomposition of $G$, we add the following bags. For each edge $uv \in E(G)$ such that $u \in B_x$ let $w_1^{uv}, \dots, w_k^{uv} = v$ be the path replacing $uv$ in $^kG$. We add the bags $B_x \cup \{w_i^{uv}, w_{i+1}^{uv}\}$ for $i = 1, k$ in order, one path at a time. Clearly the decomposition has width $\pw(G) + 2$ and every vertex and edge of $^kG$ is covered. Since the intermediate vertices on the paths only appear in two consecutive bags and vertex from $G$ is retained until all paths have been covered, no vertices are forgotten and then reintroduced. We find that it is a valid path decomposition and thus $\pw(^kG) \leq \pw(G) + 2$.
\end{proof}

\begin{lemma} \label{lem:ctwPres}
     Let $G$ be a graph. Then we have that $\ctw(^kG) \leq \ctw(G)$, $\ctw(_kG) \leq k\ctw(G)$ and $\ctw(_{(k)}G) \leq \ctw(G) + k - 1$.
\end{lemma}
\begin{proof}
    In each case we will show that, given a cut decomposition of width $\ctw(G)$, we can construct a decomposition that respects the given bounds.

    Let $\pi = (v_1, \dots, v_n)$ be a cut decomposition of $G$, of width $\ctw(G)$. First note that the given decomposition already gives a decomposition of $_kG$, of width $k\ctw(G)$.

    Next, we will examine the $k$-stretch $^kG$. We will show that subdividing an edge does not increase the cutwidth. By repeatedly subdividing edges we then find that the $k$-stretch does not have larger cutwidth. Let $uv \in E(G)$ such that $u < v$ in a given cut decomposition $\pi = (v_1, \dots, v_n)$ of $G$. We create a new graph $G'$ from $G$, by adding a vertex $w$, edges $uw$ and $wv$, and removing the edge $uv$. We construct a cut decomposition $\pi'$ of $G'$ as follows. If $v_i, v_j \in V(G)$, such that $i < j$ (in $\pi$), then we also set $v_i < v_j$ in $\pi'$. For $w$, we set $w < v_i$ if $u < v_i$ in $\pi$ and $v_i < w$ otherwise. Note that for any cut of the decomposition, we have one of the following situations. (i) The cut appears before $u$ or after $v$, in which case the cut contains the same edges in both decompositions. (ii) The cut appears\footnote{In this case we use the convention that the cut immediately before and the cut immediately after $w$ get associated with the same cut $\pi$.} between $u$ and $v$, in which case it contains $uv$ in $\pi$ and either $uw$ or $wv$ in $\pi'$, but not both. In either case we find that the cut has not increased in width in $\pi'$ and thus the decomposition has the same width.

    Finally we examine the insulated $k$-thickening. As seen before we can subdivide edges without increasing the cutwidth. We will first create the $3$-stretch of $G$, by subdividing each edge twice. We will take special care to fully subdivide an edge before moving on to the next one, so that the two new vertices on the edge appear next to each other in the cut decomposition. We then replace the middle edge of each created $3$-path with $k$ parallel edges. Since the endpoints of any such bundle of edges are next to each other in the cut decomposition, each cut contains at most one bundle and thus we increase the cutwidth by at most $k-1$.
\end{proof}

We remark that the only significant blowup is that of the cutwidth, when applying the $k$-thickening. We will therefore limit our use of this transformation as much as possible.

\subsection{Reductions}


We can now prove Theorem \ref{thm:redCurveMain}. We will split the theorem into multiple separate cases and prove each case as a separate lemma. As such the proof of Theorem \ref{thm:redCurveMain} simply consists of Lemmas \ref{lem:aNotOne}, \ref{lem:aNotZero}, \ref{lem:aIsZero} and \ref{lem:abNormOne}. Note that the last point follows from first applying Lemma \ref{lem:abNormOne} and then one of the other lemmas.

\begin{lemma}\label{lem:aNotOne}
Let $(a,b) \in \mathbb{C}^2$ be a point with $|a| \notin \{0,1\}$. Also let $T(G; x, y)$ be the Tutte polynomial of $G$ and $\alpha := (a-1)(b-1)$. There exists a polynomial time reduction from computing $T$ on $(a,b)$ for graphs of given tree-, path- or cutwidth, to computing $T$ along $H_\alpha$ for graphs with the following with parameters. The treewidth and cutwidth remain $\tw(G)$ and $\ctw(G)$ respectively. The pathwidth becomes at most $\pw(G) + 2$.
\end{lemma}
We prove this lemma using essentially the same proof as given in \cite{jaeger_vertigan_welsh_1990}. Note that in our setting we use Lemmas \ref{lem:twPres}, \ref{lem:pwPres} and \ref{lem:ctwPres} to ensure that relevant parameters are not increased by the operations we perform.
\begin{proof}
By Brylawski's tensor product formula \cite{Brylawski1980}, we find the following expression for the $k$-stretch of the graph G
\begin{align} \label{eq:kStretch}
(1 + a + \dots + a^{k-1})^{k(E)}T\left(G; a^k, \frac{b + a + \dots + a^{k-1}}{1 + a + \dots + a^{k-1}}\right) = T(^kG; a, b).
\end{align}
Note that
\[
a^k - 1 = (1 + a + \cdots + a^{k-1})(a-1)
\]
and
\[
\frac{b + a + \dots + a^{k-1}}{1 + a + \dots + a^{k-1}} - 1 = \frac{b -1}{1 + a + \dots + a^{k-1}}.
\]
We find that the point on which we evaluate $T(G)$ in \eqref{eq:kStretch} also lies on $H_\alpha$.

By examining the formula for the Tutte polynomial, we find that for $n = |V(G)|$ the degree of the Tutte polynomial is at most $n^2 + n$. By choosing $k = 0, \dots, n^2 + n$, since $|a| \notin \{0,1\}$, we can find $T(G; x,y)$, for $n^2 + n+1$ different values of $(x, y) \in H_\alpha$. By lemma \ref{lem:IntPol}, we can now interpolate the univariate restriction
\[
T_\alpha(G; t) = T\left(G; \frac{\alpha}{t} + 1, t + 1\right).
\]
of $T(G)$ along $H_\alpha$.

Note that by Lemmas \ref{lem:twPres} and \ref{lem:ctwPres} the $k$-stretch preserves both the cutwidth and the treewidth of the graph and by Lemma \ref{lem:pwPres} the pathwidth increases by a constant factor. We find that any fine-grained parameterized lower bound for $H_\alpha$ extends to points $(a,b)$. 
\end{proof}

The next lemma is proven in a similar way, however it takes a bit more effort to make the numbers line up.

\begin{lemma} \label{lem:aNotZero}
Let $(a,b) \in \mathbb{C}^2$ be a point with $|b| \notin \{0,1\}$ and $a \neq 0$. Also let $T(G; x, y)$ be the Tutte polynomial of $G$ and $\alpha := (a-1)(b-1)$. There exists a polynomial time reduction from computing $T$ on $(a,b)$ for graphs of given tree-, path- or cutwidth, to computing $T$ along $H_\alpha$ for graphs with the following with parameters. The treewidth remains $\tw(G)$. The cutwidth and pathwidth become at most and $\ctw(G) + 2$ and $\pw(G) + 2$ respectively.
\end{lemma}

\begin{proof}
By Lemma \ref{lem:aNotOne} we may assume that $|a| = 1$. In the case that $a = 1$ we can still use the $k$-stretch, since
\begin{align*}
     (1 + a + \dots + a^{k-1})^{-k(E)}T(^kG; a, b) &= T\left(G; a^k, \frac{b + a + \dots + a^{k-1}}{1 + a + \dots + a^{k-1}}\right) \\
     &= T\left(G; 1, \frac{b + k - 1}{k}\right) \\
     &= T\left(G; 1, \frac{b - 1}{k} + 1\right).
\end{align*}

Where the first equality is \eqref{eq:kStretch}. Since $b \neq 1$ we can find arbitrarily many points on the curve $H_0^x$ this way. By lemma \ref{lem:IntPol} we can interpolate to find the $T(G)$ on the whole curve.

In the remaining case, i.e. $0 \neq a \neq 1$, we use the insulated $k$-thickening.

This results in the following transformation.
\[
T(_{(k)}G; a, b) = ((a+1)(1 + b + \dots + b^{k-1}) + a^2)^{k(E)}(1 + b + \dots + b^{k-1})^{|V| - k(E)}T\left(G; A, B \right)
\]
where
\begin{align*}
    A &= a^2\left(1 + \frac{a-1}{1 + b + \dots + b^{k-1}}\right) = a^2\left(\frac{a + b + \dots + b^{k-1}}{1 + b + \dots + b^{k-1}}\right) \\
    B &= 1 + \frac{b^k - 1}{(a + 1)(1 + b + \dots + b^{k-1}) + a^2}.
\end{align*}
Which allows us to move to a point with $|A| \notin \{0,1\}$, assuming that $|b| \notin \{ 0,1\}$ and $1 \neq a \neq 0$. Note that by Lemma~\ref{lem:ctwPres} this transformation only increases the cutwidth by an additive factor of $k-1$.

It is not too dificult to see that it suffices to take either $k = 2$ or $k = 3$. Suppose that $k = 2$ does not work, then we have
\[
\left|\frac{a + b}{1 + b}\right| = 1.
\]
From this we can deduce that $b = c \cdot a^{1/2}$ for some $c \in \mathbb{R}$. Now suppose that in this case $k = 3$ also does not work. We then find that
\[
\left|\frac{a + c\cdot a^{1/2} + c^2\cdot a}{1 + c\cdot a^{1/2} + c^2\cdot a}\right| = \left|\frac{a^{1/2} + c + c^2\cdot a^{1/2}}{1 + c\cdot a^{1/2} + c^2\cdot a}\right| = 1.
\]
By squaring this term and simplifying the resulting equations we find $a = 1$ and note that this case was handled previously.
\end{proof}

In the case that $a = 0$ (and $b \neq 0$), we first use the $2$-thickening to compute
\[
(1 + b)^{|V| - k(E)}T\left(G; \frac{a + b}{1 + b}, b^2\right) = T(_2G; a, b).
\]
and then continue with other transformations. Note that this approach increases the cutwidth by a factor of $2$ and thus for any lower bound $f(\ctw)$ we would get on the curve $H_\alpha$, we find a lower bound of $f(\ctw/2)$ for the point $(0,1 - \alpha)$. We summarize this in the following Lemma.
\begin{lemma} \label{lem:aIsZero}
Let $(a,b) \in \mathbb{C}^2$ be a point with $|b| \notin \{0,1\}$ and $a = 0$. Also let $T(G; x, y)$ be the Tutte polynomial of $G$ and $\alpha := (a-1)(b-1)$. There exists a polynomial time reduction from computing $T$ on $(a,b)$ for graphs of given tree-, path- or cutwidth, to computing $T$ along $H_\alpha$ for graphs with the following with parameters. The treewidth remains $\tw(G)$. The pathwidth becomes at most $\pw(G) + 2$ and the cutwidth becomes at most $2\ctw(G)$.
\end{lemma}

The remaining case concerns points where $|a|, |b| \in \{0,1\}$. We show that we can reduce non-special points of this type to a point that is covered by one of the previous lemmas.
\begin{lemma} \label{lem:abNormOne}
Let $(a,b) \in \mathbb{C}^2$, such that $|a|, |b| \in \{0,1\}$. If $(a,b)$ is not one of the $8$ special points or on $H_1$, then there exists some transformation $f$ and a computable function $g$, such that $g(a,b) \neq 0$ and
\[
T(f(G); a, b) = g(a,b)T(G; a', b')
\]
where either $|a'| \notin \{0,1\}$ or $|b'| \notin \{0,1\}$ and such that $\tw(f(G)) \leq \tw(G)$, ${\pw(f(G)) \leq \pw(G) + 2}$ and $\ctw(f(G)) \leq 6\ctw(G)$. 
\end{lemma}

\begin{proof}
We adapt a proof due to \cite{jaeger_vertigan_welsh_1990}.

Suppose that for every such transformation $f$ we have either $|a'|, |b'| \in \{0,1\}$ or $a'$ and $b'$ are not well-defined. We will show that in this case $(a,b)$ must be one of the $8$ special points or on $H_1$.

First assume that $|a| = |b| = 1$. Note that applying the $2$-stretch gives
\[
(a', b') = \left(a^2, \frac{b+a}{1+a}\right).
\]
By assumption, we have either $a = -1$ or
\[
\left|\frac{b+a}{1+a}\right| \in \{0,1\},
\]
which implies $b = -a$, $b = a^2$ or $b = 1$.
Using the $2$-thickening we find $b = -1$, $a = -b$, $a = b^2$ or $a = 1$. This reduces the list of possible points to
\begin{align*}
    (a,b) \in 
    &\{ (1,1), (-1,-1), (j, j^2), (j^2, j), (-1, i), (-1, -i), \\
    &(i, -1), (-i, -1)\} \cup \{ (a, -a): |a| = 1 \}
\end{align*}
This list can be further reduced by applying the $3$-stretch and $3$-thickening to find that only
\begin{align*}
    (a,b) \in 
    &\{ (1,1), (-1,-1), (j, j^2), (j^2, j), (-i, i), (i, -i) \}
\end{align*}
remain.

Now suppose that $|b| = 0$. Note that, again applying the 2-stretch gives
\[
(a', b') = \left(a^2, \frac{a}{1+a}\right).
\]
By assumption we have either $a = -1$, $a = 0$ or $|1 + a| = 1$. In the last case we must have $a \in \{0, j, j^2\}$. If $a \in \{0, j, j^2\}$, then $(a', b') \in \{(j, -j), (j^2, -j^2\}$ and we may apply a $3$-stretch or $3$-thickening by the previous case. We find that the only points left are $(-1,0)$ and $(0,0)$, which lies on $H_1$.

Using the $2$-thickening and a further $3$-stretch or $3$-thickening we find that the only points of the form $(0,b)$ are $(0,0)$, $(0,-1)$.

Note that the worst blowup in the cutwidth occurs when we apply a $2$-thickening, followed by a $3$-thickening, which effectively results in a $6$-thickening and thus a multiplicative blowup in the cutwidth of $6$.
\end{proof}

%% file: Forests.tex
\section{Counting forests}
\label{sec:forests}
In this section we consider the problem of counting the number of forests in a graph. This problem corresponds to the point $(2,1)$ and thus by Theorem~\ref{thm:redCurveMain} any bounds found for this problem can be lifted to the whole curve $H_0^y$.

We trivially get the following lower bound from existing bounds on the non-parameterized version of the problem~\cite{Dell_2014}.
\begin{theorem} \label{thm:ForLow}
    Computing the Tutte polynomial along the curve $H_0^y$ cannot be done in time $2^{o(\ctw(G))} n^{O(1)}$, unless \textsc{\#ETH} fails.
\end{theorem}

To complement this lower bound, we give an algorithm to count the number of forests in a graph $G$ in $c^{\tw(G)}$ time. The algorithm uses a rank based approach, the runtime of which depends on the rank of the so called \emph{forest compatibility matrix}. We start by introducing this matrix and examining its rank.
\subsection{Notation}

We will use the notation $[n] = \{1, \dots, n\}$. Unless stated otherwise, we will assume the set $[n]$ to be ordered. For sets $A, B \subseteq [n]$, we will write $A < B$ to indicate that $a < b$ for all $a \in A$ and $b \in B$.

We write $\pi \vdash S$ to indicate that $\pi$ is a partition of $S$. We write $\pi|_S$ for the partition given by restricting elements of $\pi$ to the set $S \subseteq [n]$. Given two partitions $\pi_1 \vdash S$ and $\pi_2 \vdash S$, we say that $\pi_1$ is coarser than $\pi_2$, written $\pi_1 \geq \pi_2$, if every element of $\pi_2$ is subset of on element of $\pi_1$. Given two partitions $\pi \vdash S$ and $\rho \vdash S'$ we define the join $\pi \sqcup \rho \vdash S \cup S'$ of the partitions as the finest partition of $S \cup S'$ such that both $(\pi \sqcup \rho)|_S \geq \pi$ and $(\pi \sqcup \rho)|_{S'} \geq \rho$. Intuitively put $\pi$ and $\rho$ together and merge and overlapping elements.

We will consider matrices indexed by partitions. We will write $M[\pi, \rho]$ for the element in the row corresponding to $\pi$ and the column corresponding to $\rho$. We will write $M[\pi]$ for the vector containing all elements in the row corresponding to $\pi$.

\subsection{Rank bound}

In this section we prove the following theorem, for the so called  \emph{forest compatibility matrix} $F_n$.

\begin{theorem} \label{thm:ForRank}
The rank of $F_n$ is at most $C_n$, the $n^{\text{th}}$ Catalan number. In particular $\rank(F_n) = O(4^n n^{-3/2})$
\end{theorem}

Before we can define the forest compatibility matrix, we first need the following definitions.

\begin{definition}
    We say that a boundaried graph $G = ([n] \cup V, E)$, with boundary $[n]$, is a \emph{representative forest} for a partition $\pi \vdash [n]$, if for every $S \in \pi$ there is some connected component $C \subseteq V(G)$ such that $C \cap [n] = S$. 
    
    Given two boundaried graphs $G$ and $H$, both with boundary $B$, we define the \emph{glue} $G \oplus H$ of $G$ and $H$ as follows. First take the disjoint union of $G$ and $H$. Then identify each $v \in B$ in $G$ with its analogue in $H$.
\end{definition}

This definition shows how one can relate forests and partitions. Throughout the section we will mostly consider partitions as they capture only the information we need. The following definition elaborates on this by lifting the concept of cycles in a clue of two trees to a cycle inducet by two partitions.

\begin{definition} \label{def:RepFor}
    Let $\pi, \rho \vdash [n]$ and let $G_\pi$ and $G_\rho$ be representative forests of $\pi$ and $\rho$ respectively. We say that $\pi$ and $\rho$ \emph{induce a cycle} if $G_\pi \oplus G_\rho$ contains a cycle.
\end{definition}
It is not hard to see that it does not matter which representatives $G_\pi$ and $G_\rho$ we choose, since one only needs to know the connected components on $[n]$. This means that this definition is indeed well-defined. For this same reason, in the following definition, we only need a row and column for each partition of the separator.
\begin{definition}
    We define the \emph{forest compatibility matrix} $F_S$ of a set $S$ by
    \[
        F_S[\pi, \rho] := 
            \begin{cases}
                0   & \text{if $\pi$ and $\rho$ induce a cycle} \\
                1   & \text{otherwise}
            \end{cases}
    \]
    for any $\pi, \rho \vdash S$. We will write $F_n := F_{[n]}$.
\end{definition}

Finally we will need the following definition to bound the rank of the forest compatibility matrix.

\begin{definition}
    We say that two sets $A, B \in \pi$ are \emph{crossing} on an ordering $<$, if there are ${a_1, a_2 \in A}$ and $b_1, b_2 \in B$ such that $a_1 < b_1 < a_2 < b_2$ or $b_1 < a_1 < b_2 < a_2$. If a partition contains two crossing sets, we refer to it as a \emph{crossing partition}.
\end{definition}

Throughout this section it will sometimes be convenient to think of the ordering as a permutation. 

The general idea behind the proof of Theorem \ref{thm:ForRank} is to show that any partition can be `uncrossed', i.e. its row in $F_n$ can be written as a linear combination of rows, corresponding to non-crossing partitions. 

\subsubsection{Manipulating partitions}
For the proof of Theorem \ref{thm:ForRank} we will need the following operations, which will allow us to manipulate partitions by contracting an expanding intervals and projecting down to subsets of the ground set.

\begin{definition}
    An \emph{interval} is a subset $I \subseteq [n]$ of consecutive numbers, i.e. there is no $b \notin I$ such that $a_1 < b < a_2$ for some $a_1, a_2 \in I$.
    Given an interval $I$ and a partition $\pi$ of $[n]$, we define the \emph{contraction} $\pi -_i I$ of $\pi$ by $I$ as the partition of the set $[n] -_i I := ([n]\cup \{i\})\setminus I$ given by we merging all sets that intersect $I$ and replacing $I$ by a single element $i$, i.e.
    \[
    \pi -_i I := \{S \in \pi : S \cap I = \emptyset\} \cup \left\{ \left( \bigcup\{ S \in \pi : S \cap I \neq \emptyset \} \cup \{i\} \right) \setminus I \right\}.
    \]
    If we have an ordering on $[n]$, we place $i$ in the same place in the ordering as $I$, that is for any $a \in [n] \setminus I$ and $b \in I$, we have $a < b$ if and only if $a < i$.
    
    
    We define the \emph{blowup} $\pi +_i I$ of $\pi$ by $I$ as the partition of the set $[n] +_i I := ([n]\cup I)\setminus \{i\}$, given by adding all elements of $I$ to the set that contains $i$ and then removing $i$, i.e.
    \[
    \pi +_i I := \{S \in \pi : i \notin S \} \cup \{(S \setminus \{i\}) \cup I : i \in S\}.
    \]
    Again we place $I$ in the same place in the ordering as $i$.
\end{definition}

We will sometimes abuse notation and refer to $[n] -_i I$ as simply $[n']$ for $n' = n - |I| + 1$.

We now turn our attention to a number of useful lemmas. The first lemma intuitively says that if we contract an interval contained in some partition, then any decomposition of the resulting smaller partition gives the same decomposition of the larger partition.

\begin{lemma} \label{lem:BlockUp}
    Let $\pi$ be a partition of $[n]$ and let $I$ be an interval such that $I \subseteq S \in \pi$. We set $n' = n - |I| + 1$. Suppose that for some set of partitions $\mathcal{R}$ of $[n']$, we have $F_{n'}[\pi -_i I] = \sum_{\rho \in \mathcal{R}} a_\rho F_{n'}[\rho]$. Then ${F_n[\pi] = \sum_{\rho \in \mathcal{R}} a_\rho F_n[\rho +_i I]}$.
\end{lemma}
\begin{proof}
    Let $\chi$ be some partition of $[n]$. Note that if $|S' \cap I| \geq 2$ for some $S' \in \chi$, we have that $F_n[\pi, \chi] = F_n[\rho +_i I, \chi] = 0$. Thus we may assume that $\chi$ contains no such sets. Also note that if there is some cycle that only requires $I$ and not the rest of $S$, then again we have that $F_n[\pi, \chi] = F_n[\rho +_i I, \chi] = 0$. Thus we may assume that any cycle induced by $\chi$ and $\pi$ that has a set that intersects $I$, also requires a set that intersects $S \setminus I$, but not $I$.

    We now claim that for $\chi$ with the above assumptions we have ${F_n[\rho +_i I, \chi] = F_n[\rho, \chi-_i I]}$ for any $\rho$. This would immediately imply that for such $\chi$ \[
    F_n[\pi, \chi] = F_n[\pi -_i I, \chi -_i I] = \sum_{\rho \in R} a_\rho F_n[\rho, \chi -_i I] = \sum_{\rho \in R} a_\rho F_n[\rho +_i I, \chi],
    \]
    which proves the lemma.
    
    First note that if $\rho$ and $\chi-_i I$ induce a cycle, that does not involve $i$, then $\rho +_i I$ and $\chi$ also induce that same cycle and vice versa.
    
    Now suppose that $\rho +_i I$ and $\chi$ induce a cycle involving $I$, then there is some $S'$ in the cycle that intersects $I$. By assumption there is also some set $S'' \in \chi$ in the cycle, that intersects $S \setminus I$, but not $I$. W.l.o.g. the cycle does not loop back on itself and thus these sets are the only two in the cycle that intersect $S$. Note that $S'$ gets merged into the set containing $i$, but $S''$ does not. Since the rest of the cycle does not involve $I$, it is unaffected and thus the cycle remains intact after contraction.
    
    In the reverse direction we assume that $\rho$ and $\chi-_i I$ induce a cycle involving $i$, then it is clear to see that this cycle survives after blowing up $i$, using one of the sets in $\chi$ that intersect $I$. This proves the claim and thus the lemma.
\end{proof}

This next lemma intuitively says that if we project our partition to a subset of the ground set, then any decomposition of the resulting smaller partition gives the same decomposition of the larger partition.

\begin{lemma} \label{lem:PartProj}
    Let $\pi$ be a partition of $[n]$ and let $n' < n$. Suppose that for some set of partitions $\mathcal{R}$ of $[n']$, we have $F_{n'}[\pi|_{[n']}] = \sum_{\rho \in \mathcal{R}} a_\rho F_{n'}[\rho]$, then ${F_n[\pi] = \sum_{\rho \in \mathcal{R}} a_\rho F_n[\rho \sqcup \pi|_{[n]\setminus[n']}]}$.
\end{lemma}
\begin{proof}
    Let $\chi$ be some partition of $[n]$. If $\chi$ and $\pi|_{[n]\setminus[n']}$ induce a cycle, then the statement trivially holds. In the rest of the proof we will therefore assume that any cycle induced by $\chi$ and $\pi$ requires the use of $\pi|_{[n']}$.
    
    We first define an equivalence relation $\sim$ on $[n]$ by defining two elements to be equivalent if they are either in the same set of $\chi$ or in the same set of $\pi|_{[n]\setminus[n']}$. We then complete this to a full equivalence relation. We now define the partition $\chi'$ of $[n']$ as the set of equivalence classes of $\sim$, restricted to $[n']$.

    We claim that $F_n[\rho \sqcup \pi|_{[n]\setminus[n']}, \chi] = F_{n'}[\rho, \chi']$ for any $\rho$, which would immediately imply that
    \[
    F_n[\pi, \chi] = F_{n'}[\pi|_{[n']}, \chi'] = \sum_{\rho \in \mathcal{R}} a_\rho F_{n'}[\rho, \chi'] = \sum_{\rho \in \mathcal{R}} a_\rho F_n[\rho \sqcup \pi|_{[n]\setminus[n']}, \chi]
    \]
    which proves the lemma.
    
    Suppose that $\rho \sqcup \pi|_{[n]\setminus[n']}$ and $\chi$ induce some cycle. Since the cycle must pass through $[n']$, there must be some path from one element of $[n']$ to another, induced by $\rho \sqcup \pi|_{[n]\setminus[n']}$ and $\chi$. Since all elements in this path are equivalent, this path must lie entirely inside of a set $S' \in \chi'$ and thus replacing such a path with $S'$ results in a cycle induced by $\rho$ and $\chi'$. Note that if a cycle only requires sets from $\pi|_{[n']}$, this operation results in a single set $S'$ from $\chi'$ in the new cycle. However, since any set involved in the old cycle must contain at least two elements in the path, that set together with $S'$ induces a cycle.
    
    Similarly, in the reverse direction we take a cycle induced by $\rho$ and $\chi'$ and blow up any sets of $\chi'$ into a path in the corresponding connected component to find a cycle induced by $\rho \sqcup \pi|_{[n]\setminus[n']}$ and $\chi$.
\end{proof}

The following two lemmas help ensure that our operations do not introduce new crossings. The first of the two lemmas shows us that we can safely contract an interval, so long as it is contained in a set of the partition.

\begin{lemma} \label{lem:CrossInterval}
    Let $I \subseteq [n]$ be an interval of $[n]$. Let $\pi$ be a non-crossing partition of $[n] -_i I$. Then $\pi +_i I$ is also non-crossing.
\end{lemma}
\begin{proof}
    Suppose that there are $C, D \in \pi +_i I$ that are crossing. W.l.o.g. there are $c_1, c_2 \in C$ and $d_1, d_2 \in D$ such that $c_1 < d_1 < c_2 < d_2$. Since $\pi$ is non-crossing, this crossing does not exist in $\pi$ and thus at least one of these elements is in $I$. By definition of a blowup, we must have either $I \subseteq C$ or $I \subseteq D$. Since $I$ is an interval, it then follows that exactly one of the previously mentioned elements is in $I$. We still find a crossing in $\pi$ by replacing this element by $i$. For example, if $d_1 \in I$, we find a crossing $c_1 < i < c_2 < d_2$ in $\pi$. This again contradicts the assumption that $\pi$ is non-crossing. We conclude that $\pi +_i I$ is also non-crossing.
\end{proof}

 This next lemma shows us that, in our setting, projection is safe, as long as we do not forget any elements of sets that cross one another.

\begin{lemma} \label{lem:CrossProj}
    Let $\pi \vdash [n]$ be a partition such that only $A,B \in \pi$ cross each other and all other pairs of sets in $\pi$ are non-crossing. Then for a non-crossing partition $\rho$ of $A \cup B$ we have that $\rho \cup \pi|_{[n] \setminus (A \cup B)}$ is non-crossing.
\end{lemma}
\begin{proof}
    Suppose there are sets $C, D \in \rho \sqcup \pi|_{[n] \setminus (A \cup B)}$ that cross each other. By assumption $\pi|_{[n] \setminus (A \cup B)}$ is non-crossing and thus w.l.o.g. $C \in \rho$. Also note that since $\rho$ is non-crossing, this implies that $D \in \pi|_{[n] \setminus (A \cup B)}$.
    
    Let $I$ be the interval spanned by $A \cup B$, then since $D$ crosses $C \subseteq I$, we find that $D \cap I \neq \emptyset$ and $D$ is not an interval itself. We claim that this implies that, in $\pi$, $D$ crosses either $A$ or $B$. This would contradict the assumption that the only crossings in $\pi$ are between $A$ and $B$, which would then imply the lemma.
    
    Note that $D \cap I$ cannot include either the rightmost or the leftmost element of the interval, since these must be elements of $A \cup B$. Therefore if neither $A$ nor $B$ crosses $D$, we must have that w.l.o.g. $A < D \cap I < B$. This is not possible, since $A$ and $B$ must cross at least once.
\end{proof}

\subsubsection{Proof of the rank bound}

With Lemmas \ref{lem:BlockUp}, \ref{lem:PartProj}, \ref{lem:CrossInterval} and \ref{lem:CrossProj} in hand, we are now ready to describe the main uncrossing operation.

\begin{lemma} \label{lem:PartSwap}
    Let $\pi$ be a non-crossing partition on an ordering $p$. In time $O(n)$ we can find constants $c_\rho$, such that $F_n[\pi] = \sum_{\rho \in \mathcal{N}} c_\rho F_n[\rho]$, where $\mathcal{N}$ is the set of partitions that are non crossing on $p \circ (i, i+1)$.
\end{lemma}
\begin{proof}
    Throughout the proof, we will consider the partition $\pi$ on the ordering $p \circ (i, i+1)$. We first note that since $\pi$ is non-crossing on $p$, any crossing of $\pi$ must involve both $i$ and $i+1$. Let $i \in A \in \pi$ and $i+1 \in B \in \pi$. If $A = B$, then $\pi$ is non-crossing and thus we may assume that $A \neq B$. Note that $\pi|_{A \cup B}$, when viewed as a partition of $A \cup B$, consists of either $4$ or $5$ intervals which alternate between $A$ and $B$. Define $\pi'$ as the partition given by contracting these intervals. We find that $\pi'$ is a partition on $n'$ elements, where either $n'=4$ or $n'=5$ elements, with intervals of size $1$ (see Figure~\ref{fig:Uncross}).

    \begin{figure}[h]
        \centering
        \includegraphics[width = \textwidth]{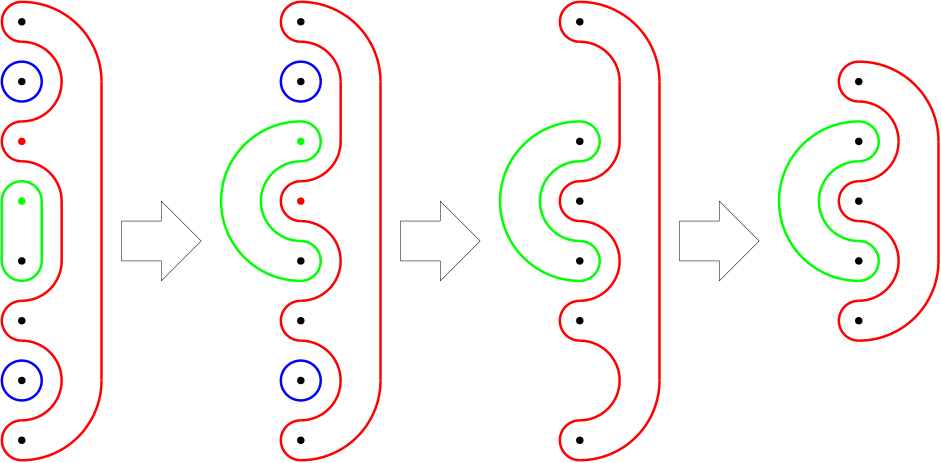}
        \caption{From left to right, these are examples of $\pi$ before the swap, $\pi$ after the swap, $\pi|_{A \cup B}$ and $\pi'$.}
        \label{fig:Uncross}
    \end{figure}
    
    We can explicitly construct the forest compatibility matrices for $n'\in \{4,5\}$ and check that the non-crossing partitions give a basis; with this manuscript we provided a MATLAB script that checks this. Thus we can write 
    \[
    F_{n'}[\pi'] = \sum_{\rho \in \mathcal{R}}c_{\rho}F_{n'}[\rho],
    \]
    where $\mathcal{R}$ is the set of non-crossing partitions of $[n']$. By Lemma \ref{lem:BlockUp} we find that
    \[
    F_{A\cup B}[\pi|_{A \cup B}] = \sum_{\rho \in \mathcal{R}}c_{\rho}F_{A\cup B}[\rho+_{i_1} I_1 + \dots +_{i_{n'}} I_{n'}].
    \]
    By Lemma \ref{lem:CrossInterval} each $\rho+_{i_1} I_1 + \dots +_{i_{n'}} I_{n'}$ is still non-crossing. By Lemma \ref{lem:PartProj} we find
    \[
    F_n[\pi] = \sum_{\rho \in \mathcal{R}}c_{\rho}F_{A\cup B}[(\rho +_{i_1} I_1 + \dots +_{i_{n'}} I_{n'}) \cup \pi|_{[n] \setminus (A \cup B)}].
    \]
    By Lemma \ref{lem:CrossProj} each $(\rho+_{i_1} I_1 + \dots +_{i_{n'}} I_{n'}) \cup \pi|_{[n] \setminus (A \cup B)}$ is still non-crossing. We conclude that $F_n[\pi]$ can be written as a linear combination of rows corresponding to non-crossing partitions.

    Note that we can construct $\pi'$ in $O(n)$ time. We then find the $c_\rho$ in $O(1)$ time and reconstruct the $(\rho+_{i_1} I_1 + \dots +_{i_{n'}} I_{n'}) \cup \pi|_{[n] \setminus (A \cup B)}$ in $O(n)$ time.
\end{proof}

By repeatedly applying Lemma \ref{lem:PartSwap}, we can prove the following theorem.

\begin{theorem} \label{thm:CompRank}
    The rows corresponding to non-crossing partitions span a row basis of the forest compatibility matrix $F_n$.
\end{theorem}
\begin{proof}
    Let $\pi$ be a partition of $[n]$ such that we can turn it into a non-crossing partition by swapping two consecutive elements $i$ and $i+1$ in the order of $[n]$. By Lemma \ref{lem:PartSwap} we can write the row $F_n[\pi]$ corresponding to $\pi$ as a linear combination of rows corresponding to non-crossing partitions of $[n]$. This shows that, for $B_p$ the set of rows corresponding to non-crossing partitions on $p$, we have $B_{p \circ (i, i+1)} \subseteq \spa(B_p)$. Since every partition is non-crossing for some permutation and every permutation can be decomposed into $2$-cycles on consecutive elements, this implies that every row can be written as a linear combination of rows corresponding to non-crossing partitions on some fixed ordering $p$.
\end{proof}

From this we immediately find a proof for  Theorem~\ref{thm:ForRank}.

\begin{proof}[Proof of Theorem \ref{thm:ForRank}]
    By Theorem~\ref{thm:CompRank} the non-crossing partitions form a basis of $F_n$. Since there are $C_n$ such partitions we find $\rank(F_n) \leq C_n$.
\end{proof}

\subsection{Algorithm}

We will now describe the algorithm for counting forests. We first define the dynamic programming table and the notion of representation. We then handle each type of node in the tree/path decomposition separately and summarize at the end.

\begin{definition}
    Let $G$ be a graph and let $(\mathbb{T}, (B_x)_{x \in V(D)})$ be a tree/path decomposition of $G$. Recall that $G_x$ is defined as the graph induced by the union of all bags, whose nodes are descendants of $x$ in $\mathbb{T}$. We define the dynamic programming table $\tau$ by
    \begin{align*}
    \tau_x[\pi] := |\{&X \subseteq E(G_x): (V,X) \text{ is acyclic }, \\&\forall u,v \in B_x \text{ there is a path in $(V, X)$ from $u$ to $v$ iff } \exists S \in \pi \text{ s.t. } u, v \in S\}|
    \end{align*}
\end{definition}
In other words, the table entry $\tau_x[\pi]$ counts the number of forests in $G_x$ whose connected components agree with $\pi$. In the rest of this section, we will refer to the number of nonzero entries $\tau_x[\pi]$ in a 'row' $\tau_x$ of the dynamic programming table as the support of $\tau_x$, written $\supp(\tau_x)$. Our aim will be to ensure that the support of our rows remains contained in the entries corresponding to non-crossing partitions for some ordering on the bag $B_x$. This is captured in the following definition.

\begin{definition}
    We say a vector $a$, indexed by partitions, is \emph{reduced on an ordering $p$}, if $a_\pi = 0$ for any partition $\pi$ that is crossing for $p$.
\end{definition}

In order to ensure that we do not lose any relevant information we will reduce our rows, while retaining the following property for the matrix $F_{B_x}$.

\begin{definition}
    Given a matrix $M$, we say that a vector $a$ \emph{$M$-represents} a vector $b$ if $Ma = Mb$.
\end{definition}

We now describe how the algorithm behaves on the various types of nodes. In each case we repeatedly apply one step of a naive dynamic programming algorithm and then reduce the table if it becomes too big. For ease of notation we will write $\pi \sim \rho$ if the partitions $\pi, \rho \vdash [n]$ are compatible, i.e. they do not induce a cycle.

\subsubsection*{Leaf node}

If $x$ is a leaf node we set $\tau_x'[\emptyset] := \tau_x[\emptyset] = 1$. We trivially find that $\tau_x'$ is reduced and $F_0$-represents $\tau_x$.

\subsubsection*{Vertex-introduce node}

\begin{lemma} \label{lem:ForVInt}
    Let $x$ be a vertex-introduce node with a child node $y$. Suppose that $\tau_y'$ is reduced and $F_{B_y}$-represents $\tau_y$. We can compute a row $\tau_x'$ that is reduced and $F_{B_x}$-represents $\tau_x$, in time $O(\rank(F_{B_x}))$.
\end{lemma}
\begin{proof}
    If $x$ is a vertex-introduce node, introducing $v$. We set
    \[
    \tau_x'[\pi \cup \{\{v\}\}] := \tau_y'[\pi]
    \]
    and
    \[
    \tau_x'[\pi] := 0
    \]
    for any $\pi$ in which $v$ does not appear as a singleton. Clearly for any non-crossing partition $\pi$, we have that $\pi \cup \{v\}$ is still non-crossing and thus $\tau_x'$ is reduced.

    Note that by definition, we need to show that $F_{B_x}\tau_x' = F_{B_x}\tau_x$. In the following derivation we show that this equality holds at the entry corresponding to any arbitrary partition $\rho \vdash B_x$.
    \begin{align*}
        \sum_{\pi \sim \rho} \tau_x[\pi]
        &= \sum_{\substack{\pi \sim \rho \\ \{v\} \in \pi}} \tau_y[\pi \setminus \{\{v\}\}] \\
        &= \sum_{\pi' \sim \rho|_{B_y}} \tau_y[\pi'] \\
        &= \sum_{\pi' \sim \rho|_{B_y}} \tau_y'[\pi'] \\
        &= \sum_{\substack{\pi \sim \rho \\ \{v\} \in \pi}} \tau_y'[\pi \setminus \{\{v\}\}] \\
        &= \sum_{\pi \sim \rho} \tau_x'[\pi]
    \end{align*}
\end{proof}

\subsubsection*{Vertex-forget node}

\begin{lemma} \label{lem:ForVFor}
    Let $x$ be a vertex-forget node with a child node $y$. Suppose that $\tau_y'$ is reduced and $F_{B_y}$-represents $\tau_y$. We can compute a row $\tau_x'$ that is reduced and $F_{B_x}$-represents $\tau_x$, in time $O(\rank(F_{B_x}))$.
\end{lemma}
\begin{proof}
    Let $x$ be a vertex-forget node, forgetting $v$. We set
    \[
    \tau_x'[\pi] := \sum_{\pi'|_{B_x} = \pi} \tau_y'[\pi']
    \]
    Clearly for any non-crossing partition $\pi'$, we have that $\pi'|_{B_x}$ is still non-crossing and thus $\tau_x'$ is reduced.

    Again we now show that $F_{B_x}\tau_x' = F_{B_x}\tau_x$, by focussing on the entry of the vector at coordinate $\rho$.
    \begin{align*}
        \sum_{\pi \sim \rho} \tau_x[\pi]
        &= \sum_{\pi \sim \rho} \sum_{\pi'|_{B_x} = \pi} \tau_y[\pi'] \\
        \intertext{Note that $\pi'$ projects down to a partition that is compatible with $\rho$ if and only if ${\pi' \sim (\rho \cup \{\{v\}\})}$. We therefore find that}
        &= \sum_{\pi' \sim (\rho \cup \{\{v\}\})} \tau_y[\pi'] \\
        &= \sum_{\pi' \sim (\rho \cup \{\{v\}\})} \tau_y'[\pi'] \\
        &= \sum_{\pi \sim \rho} \sum_{\pi'|_{B_x} = \pi} \tau_y'[\pi'] \\
        &= \sum_{\pi \sim \rho} \tau_x'[\pi]
    \end{align*}
\end{proof}

\subsubsection*{Edge-introduce node}

\begin{lemma} \label{lem:ForEdge}
    Let $x$ be an edge-introduce node with a child node $y$. Suppose that $\tau_y'$ is reduced and $F_{B_y}$-represents $\tau_y$. We can compute a row $\tau_x'$ that is reduced and $F_{B_x}$-represents $\tau_x$, in time $O(\rank(F_{B_x})|B_x|^2)$.
\end{lemma}

Before we prove this lemma, we introduce the following technical lemma. This lemma will be useful to show that representation is preserved after applying the dynamic programming step.

\begin{lemma} \label{lem:PartComm}
    Let $\pi, \chi, \rho \vdash [n]$ be partitions such that $\pi \sim \chi$ and $\rho \sim \chi$. We have that $\pi \sqcup \chi \sim \rho$ if and only if $\pi \sim \rho \sqcup \chi$.
\end{lemma}
\begin{proof}
    Recall the definition of a representative forest \ref{def:RepFor}. Let $G_\pi$, $G_\chi$ and $G_\rho$ be representative forests of $\pi$, $\chi$ and $\rho$ respectively. 
    
    Suppose that $\pi \sqcup \chi \sim \chi$. Since $\pi \sim \chi$, $G_\pi \oplus G_\chi$ is a forest. Moreover it is a representative forest of $\pi \sqcup \chi$. By the same reasoning we find that $G_\rho \oplus G_\chi$ is a representative forest of $\rho \sqcup \chi$. Since $\pi \sqcup \chi \sim \rho$, we find that $(G_\pi \oplus G_\chi) \oplus G_\rho = G_\pi \oplus (G_\chi \oplus G_\rho)$ is a forest and thus $\pi \sim \rho \sqcup \chi$.

    The reverse direction follows from a similar argument.
\end{proof}
\begin{proof}[Proof of Lemma \ref{lem:ForEdge}]
    Let $x$ be an edge-introduce node for edge $uv$. It is not hard to see that if $u$ and $v$ are adjacent in the vertex ordering of $B_x$, then $\pi \sqcup \pi_{uv}$ is non-crossing if and only if $\pi$ is non-crossing. We will aim find a $F_{B_y}$-representative $\tau_y''$ of $\tau_y'$, that is reduced on an ordering $p'$ in which $u$ and $v$ are adjacent.

    By applying Lemma \ref{lem:PartSwap} to each entry of $\tau_y'$ we can find a a $F_{B_y}$-representative of $\tau_y'$, that is reduced on $p \circ (i, i+1)$, that is we can swap two consecutive elements. Using at most $|B_y|$ of these swaps we can ensure that $u$ and $v$ are adjacent. Each such swap costs $|B_y|$ time per non-zero entry of the current vector. Since any reduced vector has at most $\rank(F_{B_y})$ non-zero entries, we find a runtime of $O(\rank(F_{B_y})|B_y|^2)$.

    We can now compute the desired $\tau_x'$. We first define
    \[
    \pi_{uv} := \{\{w\} : w \in B_y \setminus \{u,v\}\} \cup \{\{u, v\}\}
    \]
    and set
    \[
    \tau_x'[\pi] := \tau_y''[\pi] + \sum_{\pi' \sqcup \pi_{uv} = \pi} F_n[\pi', \pi_{uv}] \tau_y''[\pi']
    \]
    which is still reduced on $p'$, since $u$ and $v$ are adjacent. Finally we again show that $F_{B_x}\tau_x' = F_{B_x}\tau_x$.
    \begin{align*}
        \sum_{\pi \sim \rho} \tau_y[\pi]
        &= \sum_{\pi \sim \rho} \left( \tau_x[\pi] + \sum_{\pi' \sqcup \pi_{uv} = \pi} F_n[\pi', \pi_{uv}] \tau_x[\pi'] \right) \\
        &= \sum_{\pi \sim \rho} \left( \tau_x[\pi] \right) + \sum_{\pi \sim \rho} \left( \sum_{\pi' \sqcup \pi_{uv} = \pi} F_n[\pi', \pi_{uv}] \tau_x[\pi'] \right)
        \intertext{Since $\tau_x'$ $F_n$-represents $\tau_x$ we find}
        &= \sum_{\pi \sim \rho} \left( \tau_x'[\pi] \right) + \sum_{\pi \sim \rho} \left( \sum_{\pi' \sqcup \pi_{uv} = \pi} F_n[\pi', \pi_{uv}] \tau_x[\pi'] \right) \\
        &= \sum_{\pi \sim \rho} \left( \tau_x'[\pi] \right) + \sum_{\pi' \sqcup \pi_{uv} \sim \rho} F_n[\pi', \pi_{uv}] \tau_x[\pi']
        \intertext{If $\rho \sim \pi_{uv}$, by Lemma \ref{lem:PartComm} we find}
        &= \sum_{\pi \sim \rho} \left( \tau_x'[\pi] \right) + \sum_{\pi' \sim \rho \sqcup \pi_{uv}} \left( \tau_x[\pi'] \right) \\
        &= \sum_{\pi \sim \rho} \left( \tau_x'[\pi] \right) + \sum_{\pi' \sim \rho \sqcup \pi_{uv}} \left( \tau_x'[\pi'] \right) \\
        &= \sum_{\pi \sim \rho} \left( \tau_x'[\pi] \right) + \sum_{\pi \sim \rho} \left( \sum_{\pi' \sqcup \pi_{uv} = \pi} F_n[\pi', \pi_{uv}] \tau_x'[\pi'] \right) \\
        &= \sum_{\pi \sim \rho} \left( \tau_x'[\pi] + \sum_{\pi' \sqcup \pi_{uv} = \pi} F_n[\pi', \pi_{uv}] \tau_x'[\pi'] \right) \\
        &= \sum_{\pi \sim \rho} \tau_y'[\pi]
        \intertext{Otherwise we find $\pi' \sqcup \pi_{uv} \not\sim \rho$ for any $\pi'$ and thus}
        \sum_{\pi \sim \rho} \tau_y[\pi]
        &= \sum_{\pi \sim \rho} \left( \tau_x'[\pi] \right) \\
        &= \sum_{\pi \sim \rho} \left( \tau_x'[\pi] \right) + \sum_{\pi \sim \rho} \left( \sum_{\pi' \sqcup \pi_{uv} = \pi} F_n[\pi', \pi_{uv}] \tau_x'[\pi'] \right) \\
        &= \sum_{\pi \sim \rho} \tau_y'[\pi]
    \end{align*}
\end{proof}

\subsubsection*{Join node}

\begin{lemma} \label{lem:ForJoin}
    Let $x$ be a join node with child nodes $y_1$ and $y_2$. Suppose that $\tau_{y_i}'$ is reduced and $F_{B_{y_i}}$-represents $\tau_{y_i}$ for $i = 1, 2$. We can compute a row $\tau_x'$ that is reduced and $F_{B_x}$-represents $\tau_x$, in time $O(\rank(F_{B_x})^3|B_x|^3)$.
\end{lemma}
\begin{proof}
    We begin by setting
    \[
    \tau_x''[\pi] := \sum_{\pi_1 \sqcup \pi_2 = \pi} F_n[\pi_1, \pi_2] \tau_{y_1}'[\pi_1] \tau_{y_2}'[\pi_2]
    \]
    We will first prove that $\tau_x''$ $F_{B_x}$-represents $\tau_x$ and then reduce it afterwards.
    \begin{align*}
        \sum_{\pi \sim \rho}\tau_x[\pi]
        &= \sum_{\pi \sim \rho} \sum_{\pi_1 \sqcup \pi_2 = \pi} F_n[\pi_1, \pi_2] \tau_{y_1}[\pi_1] \tau_{y_2}[\pi_2]
        \intertext{By changing the order in which we pick $\pi$, $\pi_1$ and $\pi_2$ we can rewrite this expression as}
        &= \sum_{\pi \sim \rho} \sum_{\pi_1 \leq \pi} \tau_{y_1}[\pi_1] \sum_{\pi_1 \sqcup \pi_2 = \pi} F_n[\pi_1, \pi_2] \tau_{y_2}[\pi_2] \\
        &= \sum_{\pi_1} \tau_{y_1}[\pi_1] \sum_{\substack{\pi \sim \rho \\ \pi_1 \leq \pi}} \sum_{\pi_1 \sqcup \pi_2 = \pi} F_n[\pi_1, \pi_2] \tau_{y_2}[\pi_2]
        \intertext{We can now merge the two inner sums into one, which results in}
        &= \sum_{\pi_1} \tau_{y_1}[\pi_1] \sum_{\pi_1 \sqcup \pi_2 \sim \rho} F_n[\pi_1, \pi_2] \tau_{y_2}[\pi_2]
        \intertext{If $\rho \sim \pi_1$, by Lemma \ref{lem:PartComm} we find}
        \sum_{\pi_1 \sqcup \pi_2 \sim \rho} F_n[\pi_1, \pi_2] \tau_{y_2}[\pi_2]
        &= \sum_{\pi_2 \sim \rho \sqcup \pi_1} \tau_{y_2}[\pi_2] \\
        &= \sum_{\pi_2 \sim \rho \sqcup \pi_1} \tau_{y_2}''[\pi_2] \\
        &= \sum_{\pi_1 \sqcup \pi_2 \sim \rho} F_n[\pi_1, \pi_2] \tau_{y_2}''[\pi_2] \\
        \intertext{Otherwise we find}
        \sum_{\pi_1 \sqcup \pi_2 \sim \rho} F_n[\pi_1, \pi_2] \tau_{y_2}[\pi_2]
        &= 0 
        = \sum_{\pi_1 \sqcup \pi_2 \sim \rho} F_n[\pi_1, \pi_2] \tau_{y_2}''[\pi_2] \\
        \intertext{Either way we find}
        \sum_{\pi \sim \rho}\tau_x[\pi]
        &= \sum_{\pi_1} \tau_{y_1}[\pi_1] \sum_{\pi_1 \sqcup \pi_2 \sim \rho} F_n[\pi_1, \pi_2] \tau_{y_2}''[\pi_2]
        \intertext{By applying the same operations as before, but in reverse, we find}
        &= \sum_{\pi_1} \tau_{y_1}[\pi_1] \sum_{\substack{\pi \sim \rho \\ \pi_1 \leq \pi}} \sum_{\pi_1 \sqcup \pi_2 = \pi} F_n[\pi_1, \pi_2] \tau_{y_2}''[\pi_2] \\
        &= \sum_{\pi \sim \rho} \sum_{\pi_1 \leq \pi} \tau_{y_1}[\pi_1] \sum_{\pi_1 \sqcup \pi_2 = \pi} F_n[\pi_1, \pi_2] \tau_{y_2}''[\pi_2] \\
        &= \sum_{\pi \sim \rho} \sum_{\pi_1 \sqcup \pi_2 = \pi} F_n[\pi_1, \pi_2] \tau_{y_1}[\pi_1] \tau_{y_2}''[\pi_2]
        \intertext{By a symmetric argument as given so far, we find}
        &= \sum_{\pi \sim \rho} \sum_{\pi_1 \sqcup \pi_2 = \pi} F_n[\pi_1, \pi_2] \tau_{y_1}''[\pi_1] \tau_{y_2}''[\pi_2] \\
        &= \sum_{\pi \sim \rho} \tau_x''[\pi]
    \end{align*}

    We now describe how we reduce $\tau_x''$ to find $\tau_x'$. For each partition $\pi$ such that $\tau_x''[\pi] \neq 0$, we first determine an ordering $p'$ for which $\pi$ is non-crossing. Note that we can transform $p'$ into $p$ by performing at most $|B_x|^2$ swaps, where we swap the order of two consecutive elements. Again by applying Lemma \ref{lem:PartSwap} to each entry of $\tau_x''$ we can find a a $F_{B_x}$-representative of $e_{\pi} \cdot \tau_x''[\pi]$, that is reduced on $p' \circ (i, i+1)$.

    We perform at most $O(|B_x|^2)$ such swaps each costing at most $O(\rank(F_{B_x})|B_x|)$, since the support of the vector cannot exceed $\rank(F_{B_x})$. After we have done this for every such $\pi$, we sum the resulting vectors to find an $F_{B_x}$-representative $\tau_x'$ of $\tau_x''$. Finding the vectors takes $O(\rank(F_{B_x})|B_x|^3)$ per non-zero entry of $\tau_x''$ and thus takes $O(\supp(\tau_x'')\rank(F_{B_x})|B_x|^3)$ time in total. Summing all the vectors takes at most $O(\rank(F_{B_x})\supp(\tau_x'))$ time. Since we assumed $\tau_{y_1}'$ and $\tau_{y_2}'$ to be reduced, we find that $\supp(\tau_x'') \leq \rank(F_{B_x})^2$ and thus the algorithm runs in time $O(\rank(F_{B_x})^3|B_x|^3)$.
\end{proof}

\subsubsection*{Algorithmic results}
The previous lemmas together give the following algorithms.
\begin{theorem} \label{thm:ForPWAlg}
    There exists an algorithm that, given a graph $G$ with a path decomposition of width $\pw(G)$, computes the number of forests in the graph in time $4^{\pw(G)} n^{O(1)}$.
\end{theorem}
\begin{proof}
    W.l.o.g. we assume we are given a nice path decomposition, where the first and last nodes correspond to empty bags. As mentioned in the section of leaf nodes we can directly compute a representative solution on the first node. By applying Lemma's \ref{lem:ForVInt}, \ref{lem:ForVFor} and \ref{lem:ForEdge}, we can compute representative solutions for all nodes. The row corresponding to the last node will contain a single entry, which gives the number of forests in the graph.

    By Lemma's \ref{lem:ForVInt}, \ref{lem:ForVFor} and \ref{lem:ForEdge}, each step in the dynamic program takes at most $O(\rank(F_{\pw(G)})\pw(G)^2) = O(4^{\pw(G)}\pw(G)^{1/2})$ time. Since there are $O(n^2)$ steps we find a total running time of $O(4^{\pw(G)}\pw(G)^{1/2}n^2) = 4^{\pw(G)} n^{O(1)}$.
\end{proof}
\begin{theorem} \label{thm:ForTWAlg}
    There exists an algorithm that, given a graph $G$ with a tree decomposition of width $\tw(G)$, computes the number of forests in the graph in time $64^{\tw(G)} n^{O(1)}$.
\end{theorem}
\begin{proof}
    W.l.o.g. we assume we are given a nice tree decomposition, where the leaf  nodes correspond to empty bags. We will root this decomposition in of the leaf nodes. As mentioned in the section of leaf nodes we can directly compute a representative solution on the first node. By applying Lemma's \ref{lem:ForVInt}, \ref{lem:ForVFor}, \ref{lem:ForEdge} and \ref{lem:ForJoin}, we can compute representative solutions for all nodes. The row corresponding to the root node will contain a single entry, which gives the number of forests in the graph.

    By Lemma's \ref{lem:ForVInt}, \ref{lem:ForVFor}, \ref{lem:ForEdge} and \ref{lem:ForJoin}, each step in the dynamic program takes at most $O(\rank(F_{\tw(G)})^3\tw(G)^3) = O(64^{\tw(G)}\tw(G)^{-3/2})$ time. Since there are $O(n^4)$ steps we find a total running time of $O(64^{\tw(G)}\tw(G)^{-3/2}n^4) = 64^{\tw(G)} n^{O(1)}$.
\end{proof}

%% file: H2Lowerbound.tex
\begin{theorem} \label{thm:H2Low}
Computing the Tutte polynomial along the curve $H_2$ cannot be done in time $(2-\epsilon)^{\ctw(G)} n^{O(1)}$, unless \textsc{SETH} fails.
\end{theorem}
\begin{proof}
The Tutte polynomial on this curve specializes to the partition function of the Ising model on $G = (V,E)$ \cite{1982Baxter}. Computing this function in its entirety is equivalent to computing the generating function
\[
C_G(z) = \sum \limits _{k=0}^\infty c_kz^k
\]
of the closed subgraphs in $G$ \cite{1967Kasteleyn}. Here $c_k$ gives the the number of closed subgraphs with $k$ edges, i.e. the number of edgesets $A \subseteq E$ such that every vertex has even degree in $(V, A)$ and $|A| = k$. Computing all coefficients of $C_G$ is clearly not easier than computing the number of closed subgraphs of maximum cardinality.

We finally show that computing the number of perfect matchings reduces to computing the number of maximum closed subgraphs, using a reduction from \cite{1987Jerrum} which can be slightly altered to only increase the cutwidth by an additive factor of $2$. There is a lower bound of $2^{\ctw(G)}$ for counting perfect matchings, due to \cite{2015CurtMarx}, which finishes the proof. It remains to show that we can reduce \#\textsc{PerfectMatchings} to \#\textsc{MaximumClosedSubgraphs}, while increasing the cutwidth by at most 2.

First note that if every vertex in $G$ has odd degree, then $F$ is a perfect matching if and only if $E \setminus F$ is a maximum closed subgraph and thus the number of perfect matchings on $G$ is equal to the number of maximum closed subgraphs. We will now construct a graph $G'$ that has the same number of perfect matchings as $G$, but has only vertices with odd degree. Using the above remark we then find a reduction from \#\textsc{PerfectMatchings} to \#\textsc{MaximumClosedSubgraphs}. Also note that we can determine whether a graph has at least one perfect matching in polynomial time and thus we may assume that $G$ has at least one perfect matching.

Let $v_1, v_2, \dots, v_n$ be a cut decomposition of $G$ of width $\ctw(G)$. Now let $v_1^e, v_2^e, \dots, v_l^e$ be the vertices with even degree, in order of appearance in the cut decomposition. Since $G$ has at least one perfect matching we find that $n$ is even. Since the number of odd degree vertices in a graph is always even we also find that $l$ is even. We now connect $v_{2i-1}^e$ to $v_{2i}^e$ using a $3$-star, see figure \ref{fig:3Star}, and call the resulting graph $G'$. Note that every vertex in $G'$ has odd degree and that in a perfect matching the 'dangling' vertex $d_i$ of a $3$-star has to be matched to the center $c_i$ of the $3$-star. We find that the $3$-stars have no effect on the number of perfect matchings of the graph.

\begin{figure}[h]
    \centering
    \begin{tikzpicture}[-, scale = 1.5]
		\tikzstyle{every state} = [inner sep = 0.1mm, minimum size = 5mm]
		
		\node[state]	(D)	at (0,0)		    {$d_i$};
		\node[state]	(U)	at (-1,-0.5)		{$u$};
		\node[state]	(V)	at (1,-0.5)			{$v$};
		\node[state]	(C)	at (0,1)		    {$c_i$};
		
		\path	(D)	    edge	node	{}	(U)
						edge	node	{}	(V)
						edge	node	{}	(C);
		\draw   (U) -- ($(U)+(-0.4,0)$);
		\draw   (V) -- ($(V)+(+0.4,0)$);
				
		\end{tikzpicture}
    \caption{A $3$-star between $u$ and $v$.}
    \label{fig:3Star}
\end{figure}
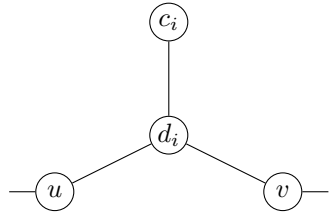

We find a new cut decomposition by simply inserting the vertices $c_i$ and $d_i$ directly after $v_{2i-1}^e$ in the cut decomposition.
\end{proof}

%% file: H2Upperbound.tex
We now prove the following matching upper bound.
\begin{theorem} \label{thm:H2Upp}
    Let $G$ be a graph with a given tree decomposition of width $\tw(G)$. There exists an algorithm that computes $T(G; a, b)$, for $(a,b) \in H_2$, in time $2^{\tw(G)}n^{O(1)}$.
\end{theorem}
\begin{proof}
    As mentioned in the proof of Theorem \ref{thm:H2Low}, computing the Tutte polynomial along the curve is equivalent to computing the partition function of the Ising model \cite{1982Baxter}. Computing this function in its entirety is equivalent to computing the generating function
    \[
    C_G(z) = \sum \limits _{k=0}^\infty c_kz^k,
    \]
    where $c_k$ gives the the number of closed subgraphs with $k$ edges \cite{1967Kasteleyn}. We can compute the coefficients of this polynomial by computing the following dynamic programming table.

    Let $S[x, p, k]$ be the number of edgesets $A$ of size $|A| = k$ in the graph below bag $B_x$ such that $deg_{G[A]}(v) \equiv_2 p(v)$. Note that the number of entries in the table is $2^{\tw(G)} n^{O(1)}$, since we only need to consider $p \in \{0,1\}^{\tw(G)}$ and $k \in [n^2]$. We may compute new entries as follows where $\emptyset$ denotes the empty vector. If $B_x$ is a leaf bag then $S[x, \emptyset, 0] = 1$ and $S[x, \emptyset, k] = 0$ otherwise. If $x$ is a vertex-forget node for vertex $v$, then
    \[
    S[x, p, k] = S[y, p', k],
    \]
    where $p'$ is the vector given by $p'(u) = p(u)$ for $u \neq v$ and $p'(v) = 0$ and $y$ is the child node of $x$. If $x$ is a vertex-introduce node for vertex $v$, then
    \[
    S[x, p, k] = 
    \begin{cases}
        S[y, p_{B_y}, k]    & \text{if } $p(v) = 0$ \\
        0                   & \text{otherwise}
    \end{cases},
    \]
    where $y$ is the child node of $x$. If $x$ is an edge-introduce node for edge $e$, then
    \[
    S[x, p, k] = S[y, p, k] + S[y, p +_2 \mathbf{1}_e, k + 1],
    \]
    where $\bold{1}_e(u) = 1$ if and only if $u$ is one of the endpoints of $e$ and where we use $+_2$ to indicate addition in $\mathbb{Z}_2^\ell$ for some $\ell$. If $x$ is a join node, with children $y_1$ and $y_2$, we use fast subset convolution as described in \cite[Theorem A.6]{2011CNPPRW} to compute
    \[
    S[x, p, k] = \sum\limits_{i = 0}^k \sum\limits_{p_1 +_2 p_2 = p} S[y_1, p_1, i]S[y_2, p_2, k-i]. 
    \]
\end{proof}

%% file: H0XLowerbound.tex
\begin{theorem} \label{thm:H0Low}
Let $0 < \alpha < 1$. Computing the Tutte polynomial along the curve $H_0^x$ cannot be done in time $(\alpha\ctw(G)-\epsilon)^{(1-\alpha)\ctw(G)/2} n^{O(1)}$, unless \textsc{SETH} fails.
\end{theorem}
\begin{proof}
In \cite{2022GNMS} a lower bound of $p^{\ctw(G)}$ is found for counting connected edgesets modulo $p$. In the reduction the authors reduce to counting essentially distinct $q$-coloring modulo $p$, with cutwidth $\ctw(G) + q^2$ and $p = q$. Thus we find a lower bound of $p^{\ctw(G) - p^2} = (\alpha \ctw(G))^{(1-\alpha)\ctw(G)/2}$ for $p = (\alpha \ctw(G))^{1/2}$.
\end{proof}

%% file: HqLowerbound.tex
\begin{theorem} \label{thm:HqLow}
Let $q \in \mathbb{Z}_{\geq 3}$. Computing the Tutte polynomial along the curve $H_q$ cannot be done in time $(q-\epsilon)^{\ctw(G)} n^{O(1)}$, unless \textsc{SETH} fails.
\end{theorem}
\begin{proof}
Note that $H_q$ contains the point $(1-q, 0)$. Computing the Tutte polynomial on this point is equivalent to counting the number of $q$-colorings of the graph $G$.

By choosing a modulus $p > q$ we can apply the results from \cite{2022GNMS} to find a lower bound of $q^{\ctw(G)}$ on the time complexity of counting $q$-colorings modulo $p$. This lower bound clearly extends to general counting.
\end{proof}

%% file: HqUpperbound.tex
\begin{theorem} \label{thm:HqUpp}
    Let $G$ be a graph with a given tree decomposition of width $\tw(G)$ and $q \in \mathbb{Z}_{\geq 3}$. There exists an algorithm that computes $T(G; a, b)$ for $(a,b) \in H_q$ in time $q^{\tw(G)} n^{O(1)}$.
\end{theorem}
This theorem is a direct consequence of combining Theorem~\ref{thm:redCurveMain} with the following folklore result:
\begin{theorem}[Folklore]
    Let $G$ be a graph with a given tree decomposition of width $\tw(G)$ and $q \in \mathbb{Z}_{\geq 3}$. There exists an algorithm that computes the number of $q$-colorings of $G$ in time $q^{\tw(G)} n^{O(1)}$.
\end{theorem}


%% file: H-qLowerbound.tex
\begin{theorem} \label{thm:H-qLow}
Let $q \in \mathbb{Z}_{>0}$. Computing the Tutte polynomial along the curve $H_{-q}$ cannot be done in $\ctw(G)^{o(\ctw(G))}$ time, unless \textsc{ETH} fails.
\end{theorem}
\begin{proof}
Like mentioned earlier $H_{-q}$ contains the point $(1+q, 0)$. For a prime $p > q$ we have that $T(G; 1+q, 0) \equiv_p T(G; 1+q - p, 0)$. This means that computing the Tutte polynomial modulo $p$ at the point $(1+q, 0)$ is equivalent to counting the number of $p - q$-colorings of $G$ modulo $p$. Since $q > 0$ and $p > q$ we find that $0 < p - q < p$ and thus as before, by \cite{2022GNMS}, we find a lower bound of $(p-q)^{\ctw(G)}$. Since the cutwidth of the construction in \cite{2022GNMS} is $O(n + rp^{r+2})$ for some $r$ dependant on $p-q$ and $\epsilon$. We find that there is no algorithm running in time $O((p - q - \epsilon)^{\ctw(G) - rp^{r+2}}) = O((\alpha\ctw(G) - \epsilon)^{\ctw(G) (1-\alpha)/(r+2)})$, where $p-q = (\alpha \ctw(G))^{1/(r+2)}$.
\end{proof}

%% file: GeneralAlgorithm.tex
\section{A general algorithm}
In this section we show how we can exploit bounded treewidth to compute the Tutte polynomial at any point in the plane, in FPT-time. For this we use a standard dynamic programming approach. 

A linear (in the input size) time FPT-algorithm has previously been given by Noble \cite{DBLP:journals/cpc/Noble98}. This algorithm is double exponential in the treewidth, where the algorithm we give here has a running time of the form $2^{O(\tw(G) \log(\tw(G))} n^{O(1)}$. We consider this an improvement for our purposes, since we are mainly interested in the dependence on the treewidth.

\begin{theorem} \label{thm:GenAlg}
There is an algorithm that, given a graph $G$ and a point $(a,b)$, computes $T(G;a,b)$ in time $\tw(G)^{O(\tw(G))} n^{O(1)}$.
\end{theorem}
\begin{proof}
Note that, in order to compute the Tutte polynomial, we only need to know the number $c_{i,j}$ of edgesets with $i$ components and $j$ edges, for $i,j = 1, \dots, n$. We can then compute
\[
T(G;a,b) = \sum\limits_{i,j=1}^n c_{i,j}(a-1)^{i-k(E)} (b-1)^{i + j - |V|},
\]
in polynomial time.

We will now focus on computing the values of $c_{i,j}$. Let $(R, (B_x)_{x \in V(R)})$ be a rooted, nice tree decomposition with root $r$. We define $C_x(\pi,i,j)$ as the number of edgesets of the graph covered by the subtree rooted\footnote{I.e. all vertices that are in some bag $y$, such that any path in $R$ from $y$ to $r$ must pass through $x$.} at bag $B_x$, with $i$ components $j$ edges and whose connected components give the partition $\pi$ on $B_x$. At the leaves of the decomposition we have $B_x = \empty$ and thus
\[
C_x(\pi,i,j) = 
\begin{cases}
1,   &\text{if } \pi = \emptyset, i = j = 0, \\
0,   &\text{otherwise}.
\end{cases}
\]

If $x$ is an introduce node for vertex $v$ with child $y$. For $N_\pi(v)$ the set of vertices in the same set of $\pi$, we have
\[
C_x(\pi,i,j) = 
\begin{cases}
\displaystyle\sum_{\emptyset \neq A \subseteq N_\pi(v)} C_y(\pi - v,i - |A|,j),   &\text{if } N_\pi(v) \neq \emptyset, \\
C_y(\pi - v,i,j - 1),    &\text{otherwise}.
\end{cases}
\]

If $x$ is a forget node for vertex $v$ with child $y$, we have
\[
C_x(\pi,i,j) = \sum_{\pi' \in 2^{B_y}, \pi'|_{B_x} = \pi} C_y(\pi',i,j)
\]

If $x$ is a join node with children $y$ and $z$, we have
\[
C_x(\pi,i,j) = \sum\limits_{k = 0}^i \sum\limits_{l = 0}^j \sum \limits_{\pi' \sqcup \pi'' = \pi} C_y(\pi,k,l)C_z(\pi,i-k,j-l),
\]
where $\pi' \sqcup \pi'' = \pi$ indicates that merging any overlapping sets in $\pi'$ and $\pi''$ results in $\pi$.
\end{proof}

%% file: Conclusion.tex
\section{Conclusion}

In this paper we gave a classification of the complexity, parameterized by treewidth/pathwidth/cutwidth, of evaluating the Tutte polynomial at integer points into either computable
\begin{itemize}
\item in polynomial time,
\item in $\tw^{O(\tw)}n^{O(1)}$ time but not in $\ctw^{o(\ctw)}n^{O(1)}$ time,
\item in $q^{\tw}n^{O(1)}$ time but not in $2^{o(\ctw)}$ (and for many points not even in $r^{\ctw}n^{O(1)}$ time for some constants $q > r$),
\end{itemize}
assuming the (Strong) Exponential Time Hypothesis.

This classification turned out to be somewhat surprising, especially considering the asymmetry between $H_0^x = \{(x,y) : x = 1 \}$ and $H_0^y = \{(x,y) : y = 1 \}$ that does not show up in other classifications such as the ones from~\cite{BrandDR19,Dell_2014,jaeger_vertigan_welsh_1990}.

Our paper leaves ample opportunities for further research.
First, we believe that our rank upper bound should have more applications for counting forests with different properties.
For example, it seems plausible that it can be used to count all Feedback Vertex Sets in time $2^{O(\tw)}n^{O(1)}$ or the number of spanning trees with $k$ components in time $2^{O(\tw)}n^{O(1)}$. The latter result would improve over a result by Peng and Fei Wan~\cite{Peng-FeiWan} that show how to count the number of spanning forests with $k$ components (or equivalently, $n-k-1$ edges) in $\tw^{O(\tw)}n^{O(1)}$ time. 
We decided to not initiate this study in this paper to retain the focus on the Tutte polynomial.

Second, it would be interesting to see if our classification of the complexity of all points on $\mathbb{Z}^2$ can be extended to a classification of the complexity of all points on $\mathbb{R}^2$ (or even $\mathbb{C}^2$). Typically, evaluation at non-integer points can be reduced to integers points (leading to hardness for non-integer points), but we were not able to establish such a reduction without considerably increasing the width parameters.

Third, it would be interesting to see if a similar classification can be made when parameterized by the vertex cover number instead of treewidth/pathwidth/cutwidth. We already know that the runtime of $2^{n}n^{O(1)}$ by Bj\"orklund et al.~\cite{BjorklundHKK08} for evaluating the Tutte polynomial cannot be strengthened to a general $2^{O(k)}n^{O(1)}$ time algorithm where $k$ is the minimum vertex cover size of the input graph due to a result by Jaffke and Jansen~\cite{JaffkeJ23}, but this still leaves the complexity of evaluating at many other points open.


%